\documentclass[submission,copyright,creativecommons]{eptcs}

\newcommand{\categoryname}[1]{{\sf{#1}}}  

\input xy  
\xyoption{all}     

\usepackage{tikz-cd}
\usepackage[nosort]{cite}
\usepackage{amssymb,amsmath,amsthm}
\usepackage{multicol}
\usepackage{tikz-cd}
\usepackage{mathtools}
\usepackage{stmaryrd}
\usepackage{enumerate}
\usepackage{tensor}
\usepackage{mathrsfs}
\usepackage{amsmath}
\usepackage{comment}
\usepackage{tikzit}
\usepackage{tikz}
\usepackage{float}
\usepackage[english]{babel}
\usepackage{braket}
\usepackage{dcolumn}               
\usepackage{bm}
\usepackage{hyperref}
\usepackage{slashed}
\usepackage{appendix}
\usepackage{comment}
\usepackage{array}
\usepackage{amstext}
\usepackage{array}
\newcolumntype{P}[1]{>{\centering\arraybackslash}p{#1}}
\newcolumntype{M}[1]{>{\centering\arraybackslash}m{#1}}

\tikzstyle{H}=[draw, color=black, fill={rgb:black,1;white,3}, shape=rectangle, tikzit fill=yellow]
\tikzstyle{thick}=[draw, line width=0.5mm]
\tikzstyle{Z}=[draw, fill=white, circle, scale=1, inner sep=0pt, minimum size=10pt]
\tikzstyle{X}=[draw, fill={rgb:black,1;white,3}, text=black, circle, scale=1, inner sep=0pt, minimum size=10pt, tikzit fill={rgb,255: red,191; green,191; blue,191}]
\tikzstyle{Xthick}=[draw, fill=white, circle, scale=1, inner sep=0pt, minimum size=15pt, line width=0.5mm, tikzit fill={rgb,255: red,191; green,191; blue,191}]
\tikzstyle{Zthick}=[draw, fill={rgb:black,1;white,3}, text=black, circle, scale=1, inner sep=0pt, minimum size=15pt, line width=0.5mm]
\tikzstyle{phase}=[draw, fill=white, diamond, scale=1, inner sep=0pt, minimum size=10pt]
\tikzstyle{discard}=[draw, xscale=2.2, ground, rotate=90]
\tikzstyle{mmixed}=[draw, quantum, yscale=-2.2, ground, rotate=180]
\tikzstyle{quantum}=[line width=.6mm]
\tikzstyle{map}=[draw, color=black, fill=white, rectangle]
\tikzstyle{mapperp}=[draw, color=white, fill=black, text=white, rectangle]
\tikzstyle{s}=[draw, color=black, fill=gray, rectangle]
\tikzstyle{mapthick}=[draw, color=black, fill=white, rectangle, inner sep=0pt, minimum size=15pt, line width=0.5mm]
\tikzstyle{otimes}=[draw, fill=white, rotate=45, scale=0.9, minimum height=.1cm, circle, append after command={[shorten >=\pgflinewidth, shorten <=\pgflinewidth]
(\tikzlastnode.north) edge (\tikzlastnode.south)
(\tikzlastnode.east) edge (\tikzlastnode.west)}]
\tikzstyle{dot}=[thick, fill=black, circle, scale=1, inner sep=.05cm]
\tikzstyle{oplus}=[draw, scale=0.9, minimum height=.1cm, circle, append after command={[shorten >=\pgflinewidth, shorten <=\pgflinewidth]
(\tikzlastnode.north) edge (\tikzlastnode.south)
(\tikzlastnode.east) edge (\tikzlastnode.west)
}]
\tikzstyle{andin}=[draw, and gate US, rotate=90, scale=1, fill=white, label={center:{\it \&}}]
\tikzstyle{mulin}=[draw, and gate US, rotate=90, scale=1, fill=white, label={center:{}}]
\tikzstyle{andout}=[draw, and gate US, rotate=-90, scale=1, fill=white, label={center:{\it \&}}]
\tikzstyle{scalar}=[draw, rounded corners=1ex, rectangle round south west=false, rectangle round south east=false, tikzit fill={rgb,255: red,129; green,253; blue,255}, tikzit shape=rectangle]
\tikzstyle{scalarop}=[draw, rounded corners=1ex, rectangle round north west=false, rectangle round north east=false, tikzit fill={rgb,255: red,116; green,172; blue,255}, tikzit shape=rectangle]
\tikzstyle{fanin}=[draw, shape border rotate=30, regular polygon, regular polygon sides=3, fill=white, inner sep=.1cm]
\tikzstyle{fanout}=[draw, shape border rotate=-30, regular polygon, regular polygon sides=3, fill=white, inner sep=.1cm]
\tikzstyle{onein}=[draw, shape border rotate=30, regular polygon, regular polygon sides=3, fill=black, inner sep=.04cm, scale=1.2]
\tikzstyle{oneout}=[draw, shape border rotate=-30, regular polygon, regular polygon sides=3, fill=black, inner sep=.04cm, scale=1.2]
\tikzstyle{zeroin}=[draw, shape border rotate=30, regular polygon, regular polygon sides=3, fill=white, inner sep=.04cm, scale=1.2]
\tikzstyle{zeroout}=[draw, shape border rotate=-30, regular polygon, regular polygon sides=3, fill=white, inner sep=.04cm, scale=1.2]
\tikzstyle{red}=[fill=red, draw=black, shape=circle, tikzit shape=circle]
\tikzstyle{green}=[fill=green, draw=black, shape=circle]
\tikzstyle{new style 0}=[fill=black, draw=black, shape=circle]
\tikzstyle{white}=[fill=white, draw=black, shape=circle]
\tikzstyle{rect k}=[fill=white, draw=black, shape=rectangle]
\tikzstyle{new style 1}=[fill=yellow, draw=black, shape=rectangle]
\tikzstyle{new style 2}=[fill=blue, draw=black, shape=circle]
\tikzstyle{new style 3}=[fill=cyan, draw=blue, shape=rectangle]
\tikzstyle{new style 4}=[fill=white, draw=black, shape=circle]

\tikzstyle{new edge style 0}=[-]
\tikzstyle{boldedge}=[-, line width=1.6pt, shorten <=-0.17mm, shorten >=-0.17mm, tikzit draw=blue]

\newif\ifpgfshaperectangleroundnortheast
\newif\ifpgfshaperectangleroundnorthwest
\newif\ifpgfshaperectangleroundsoutheast
\newif\ifpgfshaperectangleroundsouthwest
\pgfkeys{/pgf/.cd,
  rectangle round north east/.is if=pgfshaperectangleroundnortheast,
  rectangle round north west/.is if=pgfshaperectangleroundnorthwest,
  rectangle round south east/.is if=pgfshaperectangleroundsoutheast,
  rectangle round south west/.is if=pgfshaperectangleroundsouthwest,
  rectangle round north east, rectangle round north west,
  rectangle round south east, rectangle round south west,
}
\makeatletter
\def\pgf@sh@bg@rectangle{%
  \pgfkeysgetvalue{/pgf/outer xsep}{\outerxsep}%
  \pgfkeysgetvalue{/pgf/outer ysep}{\outerysep}%
  \pgfpathmoveto{\pgfpointadd{\southwest}{\pgfpoint{\outerxsep}{\outerysep}}}%
  {\ifpgfshaperectangleroundnorthwest\else\pgfsetcornersarced{\pgfpointorigin}\fi%
    \pgfpathlineto{\pgfpointadd{\southwest\pgf@xa=\pgf@x\northeast\pgf@x=\pgf@xa}{\pgfpoint{\outerxsep}{-\outerysep}}}}%
  {\ifpgfshaperectangleroundnortheast\else\pgfsetcornersarced{\pgfpointorigin}\fi%
    \pgfpathlineto{\pgfpointadd{\northeast}{\pgfpoint{-\outerxsep}{-\outerysep}}}}%
  {\ifpgfshaperectangleroundsoutheast\else\pgfsetcornersarced{\pgfpointorigin}\fi%
    \pgfpathlineto{\pgfpointadd{\southwest\pgf@ya=\pgf@y\northeast\pgf@y=\pgf@ya}{\pgfpoint{-\outerxsep}{\outerysep}}}}%
  {\ifpgfshaperectangleroundsouthwest\else\pgfsetcornersarced{\pgfpointorigin}\fi%
    \pgfpathclose}}
\tikzstyle{red}=[fill=red, draw=black, shape=circle, tikzit shape=circle]
\tikzstyle{green}=[fill=green, draw=black, shape=circle]
\tikzstyle{new style 0}=[fill=black, draw=black, shape=circle]
\tikzstyle{white}=[fill=white, draw=black, shape=circle]
\tikzstyle{rect k}=[fill=white, draw=black, shape=rectangle]
\tikzstyle{new style 1}=[fill=yellow, draw=black, shape=rectangle]
\tikzstyle{H}=[draw, color=black, fill={rgb:black,1;white,3}, shape=rectangle, tikzit fill=yellow]
\tikzstyle{thick}=[draw, line width=0.5mm]
\tikzstyle{not edge}=[-, dashed, dash pattern=on 2pt off 1.5pt, thick, draw={rgb,255: red,255; green,68; blue,68}]
\tikzstyle{Z}=[draw, fill=white, circle, scale=1, inner sep=0pt, minimum size=10pt]
\tikzstyle{X}=[draw, fill={rgb:black,1;white,3}, text=black, circle, scale=1, inner sep=0pt, minimum size=10pt, tikzit fill={rgb,255: red,191; green,191; blue,191}]
\tikzstyle{Xthick}=[draw, fill=white, circle, scale=1, inner sep=0pt, minimum size=15pt, line width=0.5mm, tikzit fill={rgb,255: red,191; green,191; blue,191}]
\tikzstyle{Zthick}=[draw, fill={rgb:black,1;white,3}, text=black, circle, scale=1, inner sep=0pt, minimum size=15pt, line width=0.5mm]
\tikzstyle{phase}=[draw, fill=white, diamond, scale=1, inner sep=0pt, minimum size=10pt]
\tikzstyle{discard}=[draw, xscale=2.2, ground, rotate=90]
\tikzstyle{mmixed}=[draw, quantum, yscale=-2.2, ground, rotate=180]
\tikzstyle{quantum}=[line width=.6mm]
\tikzstyle{map}=[draw, color=black, fill=white, rectangle]
\tikzstyle{mapperp}=[draw, color=white, fill=black, text=white, rectangle]
\tikzstyle{s}=[draw, color=black, fill=gray, rectangle]
\tikzstyle{mapthick}=[draw, color=black, fill=white, rectangle, inner sep=0pt, minimum size=15pt, line width=0.5mm]
\tikzstyle{otimes}=[draw, fill=white, rotate=45, scale=0.9, minimum height=.1cm, circle, append after command={[shorten >=\pgflinewidth, shorten <=\pgflinewidth]
(\tikzlastnode.north) edge (\tikzlastnode.south)
(\tikzlastnode.east) edge (\tikzlastnode.west)}]
\tikzstyle{dot}=[thick, fill=black, circle, scale=1, inner sep=.05cm]
\tikzstyle{oplus}=[draw, scale=0.9, minimum height=.1cm, circle, append after command={[shorten >=\pgflinewidth, shorten <=\pgflinewidth]
(\tikzlastnode.north) edge (\tikzlastnode.south)
(\tikzlastnode.east) edge (\tikzlastnode.west)
}]
\tikzstyle{andin}=[draw, and gate US, rotate=90, scale=1, fill=white, label={center:{\it \&}}]
\tikzstyle{mulin}=[draw, and gate US, rotate=90, scale=1, fill=white, label={center:{}}]
\tikzstyle{andout}=[draw, and gate US, rotate=-90, scale=1, fill=white, label={center:{\it \&}}]
\tikzstyle{scalar}=[draw, rounded corners=1ex, rectangle round south west=false, rectangle round south east=false, tikzit fill={rgb,255: red,129; green,253; blue,255}, tikzit shape=rectangle]
\tikzstyle{scalarop}=[draw, rounded corners=1ex, rectangle round north west=false, rectangle round north east=false, tikzit fill={rgb,255: red,116; green,172; blue,255}, tikzit shape=rectangle]
\tikzstyle{fanin}=[draw, shape border rotate=30, regular polygon, regular polygon sides=3, fill=white, inner sep=.1cm]
\tikzstyle{fanout}=[draw, shape border rotate=-30, regular polygon, regular polygon sides=3, fill=white, inner sep=.1cm]
\tikzstyle{onein}=[draw, shape border rotate=30, regular polygon, regular polygon sides=3, fill=black, inner sep=.04cm, scale=1.2]
\tikzstyle{oneout}=[draw, shape border rotate=-30, regular polygon, regular polygon sides=3, fill=black, inner sep=.04cm, scale=1.2]
\tikzstyle{zeroin}=[draw, shape border rotate=30, regular polygon, regular polygon sides=3, fill=white, inner sep=.04cm, scale=1.2]
\tikzstyle{zeroout}=[draw, shape border rotate=-30, regular polygon, regular polygon sides=3, fill=white, inner sep=.04cm, scale=1.2]
\tikzstyle{red}=[fill=red, draw=black, shape=circle, tikzit shape=circle]
\tikzstyle{green}=[fill=green, draw=black, shape=circle]
\tikzstyle{new style 0}=[fill=black, draw=black, shape=circle]
\tikzstyle{white}=[fill=white, draw=black, shape=circle]
\tikzstyle{rect k}=[fill=white, draw=black, shape=rectangle]
\tikzstyle{new style 1}=[fill=yellow, draw=black, shape=rectangle]
\tikzstyle{new style 2}=[fill=blue, draw=black, shape=circle]
\tikzstyle{new style 3}=[fill=cyan, draw=blue, shape=rectangle]
\tikzstyle{new style 4}=[fill=white, draw=black, shape=circle]

\tikzstyle{new edge style 0}=[-]
\tikzstyle{boldedge}=[-, line width=1.6pt, shorten <=-0.17mm, shorten >=-0.17mm, tikzit draw=blue]

\tikzstyle{new edge style 0}=[-]
\tikzstyle{boldedge}=[-, line width=1.6pt, shorten <=-0.17mm, shorten >=-0.17mm, tikzit draw=blue]
\tikzstyle{new edge style 0}=[-]
\tikzstyle{strings}=[baseline={([yshift=-.5ex]current bounding
\usetikzlibrary{arrows,decorations.pathmorphing,backgrounds,positioning,fit}box.center)}]

\newtheorem{observation}{Remark}[section]
\newtheorem{example}[observation]{Example}
\newtheorem{definition}[observation]{Definition}

\newtheorem{lemma}[observation]{Lemma}
\newtheorem{proposition}[observation]{Proposition}
\newtheorem{theorem}[observation]{Theorem}
\newtheorem{remark}[observation]{Remark}
\newtheorem{corollary}[observation]{Corollary}

\title{Categories of Kirchhoff Relations}

\author{Robin Cockett
\institute{University of Calgary\\
Alberta, Canada}
\institute{Department of Computer Science}
\email{robin@ucalgary.ca}
\and
Amolak Ratan Kalra 
\institute{University of Calgary\\
Alberta, Canada}
\institute{Department of Computer Science}
\email{amolakratan.kalra@ucalgary.ca}
\and
Shiroman Prakash 
\institute{Dayalbagh Educational Institute\\
Uttar Pradesh, India}
\institute{Department of Physics and Computer Science}
\email{sprakash@dei.ac.in}
}

\begin{document}
\maketitle
\begin{abstract}
It is known that the category of affine Lagrangian relations, $\categoryname{AffLagRel_F}$, over a field, $F$, of integers modulo a prime $p$ (with $p > 2$) is isomorphic to the category of stabilizer quantum circuits for $p$-dits.  Furthermore, it is known that electrical circuits (generalized for the field $F$) occur as a natural subcategory of  $\categoryname{AffLagRel_F}$. The purpose of this paper is to provide a characterization of the relations in this subcategory -- and in important subcategories thereof -- in terms of parity-check and generator matrices as used in error detection.\\
In particular, we introduce the subcategory consisting of Kirchhoff relations to be (affinely) those Lagrangian relations that conserve total momentum or equivalently satisfy Kirchhoff's current law.  Maps in this subcategory can be generated by electrical components  (generalized for the field $F$): namely resistors, current dividers, and current and voltage sources.  The ``source'' electrical components deliver the affine nature of the maps while current dividers add an interesting quasi-stochastic aspect.\\
We characterize these Kirchhoff relations in terms of parity-check matrices and in addition, characterizes two important subcategories: the deterministic Kirchhoff relations and the lossless relations.   The category of deterministic Kirchhoff relations as electrical circuits are generated by resistors circuits. Lossless relations, which are deterministic Kirchhoff, provide exactly the basic hypergraph categorical structure of these settings.
\end{abstract}

\tableofcontents

\section{Introduction}
Categorical techniques have been used with great success to analyze, simplify and study both quantum and electrical circuits.
On the quantum circuit side these methods have met with considerable success to give new insights into certain aspects of quantum error correction \cite{PhysRevA.102.022406, Aleks1, kissinger2022phase, kalra2019demonstration} and quantum circuit simplification \cite{de2022circuit,booth2021outcome, kissinger2020reducing}.
On the other hand the category of Lagrangian relations and more generally affine relations and its associated graphical calculus called Graphical Affine Algebra (GAA) has been used to capture the behaviour of both active and passive electrical networks \cite{baez2017props, fong2016algebra, Bonchi2019GraphicalAA,gla,zanasi,boisseau2021string,bonchi2017interacting,baez2014categories}. 

\medskip

Recently a connection between Lagrangian relations and stabilizer quantum mechanics for qudits of odd prime dimensions was made by the work of Comfort and Kissinger \cite{comfort2021graphical}. Inspired by this work and earlier results in physical systems theory \cite{DPS2009, DPS2011, DPS2014,DPS1,DPS2,karunakaran1975systems} and qudit quantum computation \cite{gross2006hudson,spekkens2016quasi, Wootters1987,Gottesman1999,mitpaper, prakash2021normal, Catani_2017, PhysRevA.75.032110} we introduce the subcategory of momentum-conserving Lagrangian relations over a field $F$ called Kirchhoff relations and denoted as $\categoryname{KirRel}_{F}$. We study $\categoryname{KirRel}_{F}$ and its subcategories using parity-check and generator matrices from error correction. By studying the structure of parity-check matrices we isolate different subcategories of $\categoryname{KirRel}$. These subcategories correspond to different classes of electrical networks. In the last section of the paper we list a universal set of generators for $\categoryname{KirRel}$, all these generators have natural interpretations in terms of electrical elements namely as resistors, junctions, and current dividers. 


\section{Linear and Affine Relations}
Let $F$ be a field: we will mainly be interested in finite fields of integers modulo some prime $p>2$.  A \textbf{linear relation} $\mathcal{R}:m \to n$ between the vector spaces $F^m$ and $F^n$ is a linear subspace $\mathcal{R} \subseteq F^{n}\oplus F^{m}$.  These linear relations over the field $F$ form a prop, $\categoryname{LinRel}_{F}$, as described in \cite{Bonchi2019GraphicalAA}, where the objects are natural numbers and morphisms are linear relations with respect to relational composition and tensor product given by direct sum.  $\categoryname{LinRel}_{F}$ is also a hypergraph category, \cite{fong2019hypergraph, fong2018seven}, in two different ways first with respect to the Frobenius structure given by the copy map and its converse (here represented by black spiders) and secondly by the Frobenius structure provided by addition and its converse (here represented by white spiders).
A special case of a linear relation is a linear function, in which each ${x} \in F^m$ is related to exactly one ${y} \in F^n$. A linear function from $F^m$ to $F^n$ can be written as an $n \times m$ matrix $A$ such that ${y} = A {x}$. Linear functions form a prop by themselves, which is given by $\iota:\categoryname{Mat}_{F}\to \categoryname{LinRel}_{F}$.
While a general linear relation $\mathcal{R}: m \to n$ cannot be represented as an $n \times m$ matrix, it is possible to represent an arbitrary linear relation 
using larger matrices with $n+m$ columns, and this idea is exploited  in the theory of linear error-correcting codes, \cite{huffman2010fundamentals}.


\subsection{Generator and parity-check matrices}

Any linear subspace $\mathcal{R} : m \to n$ can be expressed as the \textit{span} of a set of $k$ independent vectors, that is as the set of all $(x,y) \in F^m \oplus F^n$ that 
can be written as
\begin{equation}
    \begin{pmatrix} x \\ y \end{pmatrix} = G^T {t}  \label{generator}
\end{equation}
for some $t \in F^k$, where $G^T: k \to m + n$ has rank $k$.\footnote{$G^T$ indicates the transpose of $G$.}  In the terminology of linear error-correcting codes, $G$  is a \textbf{generator matrix} of the linear relation $\mathcal{R}$.  Alternatively, the linear relation can be expressed as the vectors satisfying the linear equation:
\begin{equation}
    H \begin{pmatrix} x \\ y \end{pmatrix} = 0. \label{parity-check}
\end{equation}
where $H: n+m \to n+m-k $ is a matrix of rank $n+m-k$.  In other words, $\mathcal{R}$ is the kernel of $H$.  In the terminology of linear error-correcting codes, $H$ is a \textbf{parity-check matrix} for the linear relation $\mathcal{R}$.
\begin{example} Consider the function from $F^2 \to F^2$ shown below.
Note that here we use the notation of \cite{Bonchi2019GraphicalAA} where the black spiders correspond to copying and the white spiders correspond to addition.
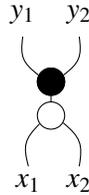
\begin{equation*}
\begin{tikzpicture}[scale=0.3]
	\begin{pgfonlayer}{nodelayer}
		\node [style=white] (0) at (1.75, -26) {};
		\node [style=none] (1) at (0.75, -28.25) {};
		\node [style=none] (2) at (3, -28.25) {};
		\node [style=white] (3) at (1.75, -24.5) {};
		\node [style=none] (4) at (0.5, -22.5) {};
		\node [style=none] (5) at (3, -22.5) {};
		\node [style=new style 0] (6) at (1.75, -24.5) {};
		\node [style=none] (7) at (0.5, -21.5) {$y_{1}$};
		\node [style=none] (8) at (3, -21.5) {$y_{2}$};
		\node [style=none] (9) at (0.75, -29) {$x_{1}$};
		\node [style=none] (10) at (3, -29) {$x_{2}$};
	\end{pgfonlayer}
	\begin{pgfonlayer}{edgelayer}
		\draw [bend right, looseness=1.25] (0) to (1.center);
		\draw [bend left, looseness=1.25] (0) to (2.center);
		\draw [bend left, looseness=1.25] (3) to (4.center);
		\draw [bend right, looseness=1.25] (3) to (5.center);
		\draw (6) to (0);
	\end{pgfonlayer}
\end{tikzpicture}
\end{equation*}
This diagram encodes the following relation $\{(x_{1},x_{2}),(y_{1},y_{2})~|~x_{1}+x_{2}=y_{1},~y_{1}=y_{2}\}$ which can be written in matrix form as:
\[
\begin{pmatrix}
1 & 1 &-1 & 0\\
0 & 0 & 1 & -1\\
\end{pmatrix}
\begin{pmatrix}
x_{1}\\
x_{2}\\
y_{1}\\
y_{2}\\
\end{pmatrix}
=\begin{pmatrix}
0\\
0\\
\end{pmatrix}
\]
Thus, $H:=\begin{pmatrix}
1 & 1 &-1 & 0\\
0 & 0 & 1 & -1\\
\end{pmatrix}$ is a parity-check matrix.   A generator matrix for the relation is: 
\[ G^{T}:=\begin{pmatrix}
1 & 0\\
0 & 1\\
1 & 1 \\
0 & 0 \\
\end{pmatrix}\]
Note that $HG^{T}=0$.
\end{example}
\begin{example}
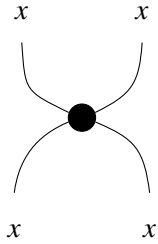
\begin{figure}
    \centering
  \begin{tikzpicture}[scale=0.4]
	\begin{pgfonlayer}{nodelayer}
		\node [style=new style 0] (0) at (12, 0.5) {};
		\node [style=none] (1) at (10, 3) {};
		\node [style=none] (2) at (14, 3) {};
		\node [style=none] (3) at (9.75, -2) {};
		\node [style=none] (4) at (14.25, -2) {};
		\node [style=none] (5) at (9.75, -3.25) {$x$};
		\node [style=none] (6) at (14.25, -3.25) {$x$};
		\node [style=none] (7) at (10, 4) {$x$};
		\node [style=none] (8) at (14, 4) {$x$};
	\end{pgfonlayer}
	\begin{pgfonlayer}{edgelayer}
		\draw [bend right, looseness=1.50] (1.center) to (0);
		\draw [bend right, looseness=1.25] (0) to (2.center);
		\draw [bend right] (0) to (3.center);
		\draw [bend left, looseness=1.25] (0) to (4.center);
	\end{pgfonlayer}
\end{tikzpicture}
    \caption{Parity check matrices for a copy spider with two inputs and two outputs}
    \label{spider2in2out}
\end{figure}
Consider the relation $\mathcal R$ between $F^2$ and $F^2$, defined by $\{((x,x),(x,x)) \in F^2\oplus F^2 | x \in F\}$, (which can be represented by a copy spider with 2 inputs and 2 outputs, as defined in the next subsection, shown in Figure \ref{spider2in2out}.) This can alternatively be specified via a parity-check matrix as the set $((x_1, x_2), (y_1, y_2)) \in F^2 \oplus F^2$ that satisfies
\begin{equation}
    \begin{pmatrix} 1 & -1 & 0 & 0 \\ 0 & 1 & -1 & 0 \\ 0 & 0 & 1 & -1 \end{pmatrix}  \begin{pmatrix} x_1 \\ x_2 \\ y_1 \\ y_2 \end{pmatrix} = 0.
\end{equation}
It can also be specified via a generator matrix as
\begin{equation}
    \begin{pmatrix} x_1 \\ x_2 \\ y_1 \\ y_2 \end{pmatrix} = \begin{pmatrix} 1 \\ 1 \\ 1 \\ 1 \end{pmatrix} \begin{pmatrix} t \end{pmatrix}.
\end{equation}
From these expressions we see that the dimension of $\mathcal R$, is $k=1$.
\end{example}
A parity-check matrix, $H$, for $\mathcal{R}$ is not unique.  Any vector that satisfies equation \eqref{parity-check}, also satisfies the equation with $H$ replaced by $H' := \beta H$, where $\beta$ is any invertible matrix with dimension $m+n-k$ (similarly, one can replace $G$ by $\beta' G$ where $\beta'$ is an invertible matrix with dimension $k$).  While parity-check matrices are not unique, whenever $H$ and $H'$ are parity-check matrices for $\mathcal{R}$ then there exists a unique comparison isomorphism $\alpha$ such that the following diagram commutes:
\begin{equation}
\begin{tikzcd}
F^{m}\oplus F^{n} \arrow[rr, "H"] \arrow[rrd, "H'"'] &  & F^{m+n-k} \arrow[d, "\alpha"] \\
                                                     &  & F^{m+n-k}                    
\end{tikzcd}
\end{equation}
This is because, in $\categoryname{Mat}_{F}$, a parity-check matrix for $\mathcal{R}$ is a universal coequalizer with kernel $\mathcal{R}$ and all such coequalizers are uniquely comparable as above. 

Finally, note that  $H$ is a parity-check matrix and $G$ is a generator matrix for the same subspace if and only if $H G^T=0$ and this is an {\em exact\/} composite in the sense that 
the image of $G^T$ is the kernel of $H$ (both of which are the subspace being classified).  Exactness of the composite is assured provided ${\sf Rank}(H)+{\sf Rank}(G)=m+n$, which 
is always the case for generator and parity check matrices for the same subspace.

As mentioned above, the parity-check and generator matrices for a relation are not unique. The rows of $G$ are linearly independent and form a basis for the linear subspace $\mathcal R$. Changing the order of rows, or making elementary row operations, such as adding one row to another row, does not change the linear subspace specified by $G$ or ${H}$. 
\begin{equation}
H =  \begin{pmatrix} 1_{m+n-k} & A  \end{pmatrix} \sigma: m+n \to m+n-k, \label{standard-formH}
\end{equation}
where $\sigma$ is a permutation matrix.  Corresponding to this is a generator matrix in the standard form:
\begin{equation}
G =  \begin{pmatrix} -A^T & 1_{k} \end{pmatrix} \sigma: m+n  \to k \label{standard-formG}
\end{equation}
where the same matrix $A$ and permutation $\sigma$ appear in equations \eqref{standard-formH} and \eqref{standard-formG} (see \cite{huffman2010fundamentals}).  
One can see that this $G$ is a generator matrix  as (using the fact that $\sigma^T = \sigma^{-1}$):
\[ H G^T = \begin{pmatrix} 1_{m+n-k} & -A  \end{pmatrix} \sigma \sigma^T \begin{pmatrix} A  \\ 1_{k} \end{pmatrix} = \begin{pmatrix} 1_{m+n-k} & -A  \end{pmatrix} \begin{pmatrix} A  \\ 1_{k} \end{pmatrix}= A - A =0 \]
and the ranks make this composite exact.  The standard form for a generator parity-check pair is, thus determined by $(A,\sigma)$.

While standard forms for characterizing a given $\mathcal{R}$ are not unique, the choice of $\sigma$ completely determines the pair $(A,\sigma)$.  This can be seen from the commuting diagram:
\begin{equation}
\begin{tikzcd}
                                              & m+n-k \arrow[d, "P"'] \arrow[rrd, "1_{m+n-k}"] &  &                          \\
                                              & m+n \arrow[rr, "H"]                    &  & m+n-k \arrow[dd, "\alpha"] \\
N \arrow[rd, "\sigma'"'] \arrow[ru, "\sigma"] &                                      &  &                          \\
                                              & m+n \arrow[rr, "H'"']                  &  & m+n-k                     
\end{tikzcd}
\end{equation}
where $P:=\begin{pmatrix} 0 \\ 1 \end{pmatrix}$.
This shows that the isomorphism $\alpha$ between parity-check matrices $\sigma H$ and $\sigma' H'$ for a relation $\mathcal{R}$ is
$\alpha=H'\sigma'\sigma^{-1} P$.  Thus, for a fixed $\sigma H$, an $\alpha$ which produces $\sigma' H'$ is completely determined by the 
permutation $\sigma'$. 
Suppose two relations $\mathcal{R}_1$ and $\mathcal{R}_2$ between the same vector spaces are described by the parity-check and generator matrices $H_1$, $G_1$ and $H_2$, $G_2$ respectively, then it is easy to see that $\mathcal{R}_1 \subset \mathcal{R}_2$ if and only if $H_2 G_1=0$. 

Two relations which can be used to convert inputs into outputs, and vice versa are:
\[ 
   \textit{cap}=\{(0,(x,x))~|~x\in F)\} ~~ \mbox{and} ~~  \textit{cup}=\{((x,x),0)~|~x\in F\}
\]
We may use cups and caps to convert any relation from $F^m$ to $F^n$ into a {\em state} (where $m=0$) or an effect (where $n=0$), or indeed into a relation between $F^{m'}$ and $F^{n'}$, where $m'+n'=m+n$.  In particular, these relations allows one to recover a relation from its parity matrix and from its generator matrix as shown below:

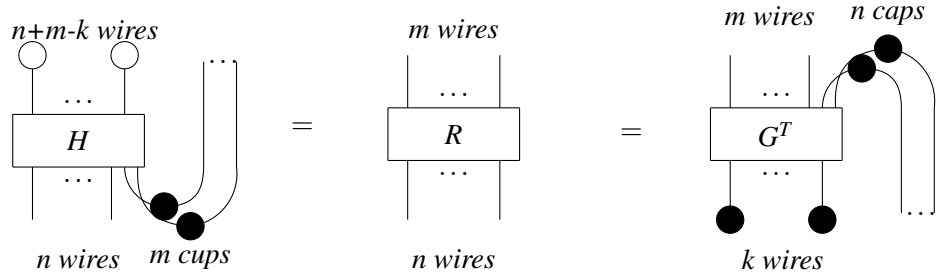
\begin{center}
\begin{tikzpicture}[scale=0.35]
	\begin{pgfonlayer}{nodelayer}
		\node [style=none] (0) at (14, 3) {};
		\node [style=none] (1) at (19, 3) {};
		\node [style=none] (2) at (14, 1) {};
		\node [style=none] (3) at (19, 1) {};
		\node [style=none] (4) at (16.5, 2) {$R$};
		\node [style=none] (5) at (14.75, 3) {};
		\node [style=none] (6) at (18.25, 3) {};
		\node [style=none] (7) at (14.75, 1) {};
		\node [style=none] (8) at (18.25, 1) {};
		\node [style=none] (9) at (14.75, 5) {};
		\node [style=none] (10) at (18.25, 5) {};
		\node [style=none] (11) at (14.75, -1.25) {};
		\node [style=none] (12) at (18.25, -1.25) {};
		\node [style=none] (13) at (16.5, -2.75) {$\textit{n wires}$};
		\node [style=none] (14) at (16.5, 0.5) {$\ldots$};
		\node [style=none] (15) at (16.5, 3.5) {$\ldots$};
		\node [style=none] (16) at (16.5, 6) {$\textit{m wires}$};
		\node [style=none] (17) at (10.75, 2.25) {$=$};
		\node [style=none] (18) at (-0.25, 0.75) {};
		\node [style=none] (19) at (4.75, 0.75) {};
		\node [style=none] (20) at (-0.25, 2.75) {};
		\node [style=none] (21) at (4.75, 2.75) {};
		\node [style=none] (22) at (2.25, 1.75) {$H$};
		\node [style=none] (23) at (0.5, 0.75) {};
		\node [style=none] (24) at (3.5, 0.75) {};
		\node [style=none] (25) at (0.5, 2.75) {};
		\node [style=none] (26) at (4, 2.75) {};
		\node [style=none] (27) at (0.5, -1.25) {};
		\node [style=none] (28) at (3.5, -1.25) {};
		\node [style=none] (29) at (0.5, 5) {};
		\node [style=none] (30) at (4, 5) {};
		\node [style=none] (31) at (2.5, 6) {$\textit{n+m-k wires}$};
		\node [style=none] (32) at (2.25, 3.25) {$\ldots$};
		\node [style=none] (33) at (2.25, 0.25) {$\ldots$};
		\node [style=none] (34) at (2.25, -2.75) {$\textit{n wires}$};
		\node [style=new style 4] (35) at (0.5, 5) {};
		\node [style=new style 4] (36) at (4, 5) {};
		\node [style=none] (37) at (4, 0.75) {};
		\node [style=none] (38) at (4.5, 0.75) {};
		\node [style=none] (39) at (7, 0.75) {};
		\node [style=none] (40) at (8.25, 0.75) {};
		\node [style=none] (41) at (5.5, -0.75) {};
		\node [style=none] (42) at (6.5, -1.5) {};
		\node [style=none] (43) at (7, 4.75) {};
		\node [style=none] (44) at (8.25, 4.75) {};
		\node [style=new style 0] (45) at (5.5, -0.75) {};
		\node [style=new style 0] (46) at (6.5, -1.5) {};
		\node [style=none] (47) at (6.5, -2.75) {$\textit{m cups}$};
		\node [style=none] (48) at (7.75, 4.75) {$\ldots$};
		\node [style=none] (49) at (26.25, 3) {};
		\node [style=none] (50) at (31.25, 3) {};
		\node [style=none] (51) at (26.25, 1) {};
		\node [style=none] (52) at (31.25, 1) {};
		\node [style=none] (53) at (28.75, 2) {$G^{T}$};
		\node [style=none] (54) at (27, 3) {};
		\node [style=none] (55) at (30, 3) {};
		\node [style=none] (56) at (27, 1) {};
		\node [style=none] (57) at (30.5, 1) {};
		\node [style=none] (58) at (27, 5) {};
		\node [style=none] (59) at (30, 5) {};
		\node [style=none] (60) at (27, -1.25) {};
		\node [style=none] (61) at (30.5, -1.25) {};
		\node [style=none] (62) at (29, -2.75) {$\textit{k wires}$};
		\node [style=none] (63) at (28.75, 0.5) {$\ldots$};
		\node [style=none] (64) at (28.75, 3.5) {$\ldots$};
		\node [style=none] (65) at (28.5, 6.5) {$\textit{m wires}$};
		\node [style=new style 0] (66) at (27, -1.25) {};
		\node [style=new style 0] (67) at (30.5, -1.25) {};
		\node [style=none] (68) at (30.5, 3) {};
		\node [style=none] (69) at (31, 3) {};
		\node [style=none] (70) at (33.5, 3) {};
		\node [style=none] (71) at (34.75, 3) {};
		\node [style=none] (72) at (32, 4.5) {};
		\node [style=none] (73) at (33, 5.25) {};
		\node [style=none] (74) at (33.5, -1) {};
		\node [style=none] (75) at (34.75, -1) {};
		\node [style=new style 0] (76) at (32, 4.5) {};
		\node [style=new style 0] (77) at (33, 5.25) {};
		\node [style=none] (78) at (33, 6.5) {$\textit{n caps}$};
		\node [style=none] (79) at (34.25, -1) {$\ldots$};
		\node [style=none] (80) at (23.25, 2) {$=$};
	\end{pgfonlayer}
	\begin{pgfonlayer}{edgelayer}
		\draw (0.center) to (1.center);
		\draw (2.center) to (3.center);
		\draw (3.center) to (1.center);
		\draw (0.center) to (2.center);
		\draw (9.center) to (5.center);
		\draw (10.center) to (6.center);
		\draw (7.center) to (11.center);
		\draw (8.center) to (12.center);
		\draw (18.center) to (19.center);
		\draw (20.center) to (21.center);
		\draw (21.center) to (19.center);
		\draw (18.center) to (20.center);
		\draw (27.center) to (23.center);
		\draw (28.center) to (24.center);
		\draw (25.center) to (29.center);
		\draw (26.center) to (30.center);
		\draw [bend right=45, looseness=1.25] (38.center) to (42.center);
		\draw [bend right=45] (42.center) to (40.center);
		\draw [bend right=45] (37.center) to (41.center);
		\draw [bend right=45, looseness=1.25] (41.center) to (39.center);
		\draw (39.center) to (43.center);
		\draw (40.center) to (44.center);
		\draw (49.center) to (50.center);
		\draw (51.center) to (52.center);
		\draw (52.center) to (50.center);
		\draw (49.center) to (51.center);
		\draw (58.center) to (54.center);
		\draw (59.center) to (55.center);
		\draw (56.center) to (60.center);
		\draw (57.center) to (61.center);
		\draw [bend left=45, looseness=1.25] (69.center) to (73.center);
		\draw [bend left=45] (73.center) to (71.center);
		\draw [bend left=45] (68.center) to (72.center);
		\draw [bend left=45, looseness=1.25] (72.center) to (70.center);
		\draw (70.center) to (74.center);
		\draw (71.center) to (75.center);
	\end{pgfonlayer}
\end{tikzpicture}
\end{center}

As described in \cite{pawelblog,comfort2021graphical} orthogonal complementation provides a functor on linear relations  
\[ (\_)^{\perp}:\categoryname{LinRel}_{F} \to \categoryname{LinRel}_{F}; \mathcal{R} \mapsto \mathcal{R}^\perp = \{ (x',y') | \forall (x,y) \in \mathcal{R}. x^Tx' -y^T y' = 0 \} \]
which, furthermore, has the effect of flipping parity-checking  and generating:

\begin{proposition}\label{ortho}
To have a parity-check (respectively generator) matrix  of $\mathcal{R}$ is precisely to have a generator (respectively parity-check) matrix of $\mathcal{R}^\perp$.
\end{proposition}

\begin{proof}
Let $\mathcal R$ be a rank $k$ relation described by generator matrix $G$. Then, for all $x \in \mathcal R^\perp$, and for all $u \in F^k$ we have $x^T G^T u=0$, by definition of $\mathcal R^\perp$. This means $x \in \mathcal R^\perp$ if and only if $Gx=0$, which means $G$ is a parity-check matrix for $\mathcal R^\perp$.
\end{proof}

This result is well-known in the coding theory literature \cite{huffman2010fundamentals} and says that a parity-check of a code $C$ is a generator matrix of the dual code $C^{\perp}$ and a generator matrix for $C$ is a parity-check matrix for the dual code $C^{\perp}$.

A useful observation describes the orthogonal of the relation of a map: Note that there is an inclusion functor $\iota:\categoryname{Mat}_{F}\to \categoryname{LinRel}_{F};~M\mapsto \{(x,y)~|~y=Mx\} $
\begin{lemma}
If $M: m\to n \in \categoryname{Mat}_{F}$ then the orthogonal complement of $\iota(M)$ is:
\[
\iota(M)^{\perp}=\{(x,y)~|~x=M^{T}y\}
\]
\label{ortho-lemma}
\end{lemma}
\begin{proof}
First note that $\dim \iota(M)^\perp=n$, so we need construct a subspace $V$ of $F^m\oplus F^n$ of dimension $n$. The subspace encoded by $\iota(M)$ is: $({x},M{x})$ for all ${x} \in F^m$. Define $V$ as $(M^T {y}, {y})$ for all ${y} \in F^n$. $V$ has dimension $n$. The inner product of any vector in $V$ with any vector in $\iota(M)$, computed as:
\begin{equation*}
    {x}^T M^T {y}-{x}^T M^T {y}=0.
\end{equation*}
Therefore $\iota(M)^{\perp}=V$.
\end{proof}


\subsection{Lagrangian Relations and Parity-check Matrices}

Symplectic forms are used in classical physics to distinguish position and momentum, and similarly, in electrical engineering, to distinguish current and voltage.
Lagrangian relations \cite{weinstein1987symplectic} are linear relations between symplectic vector spaces which, in addition, are invariant under taking their symplectic dual.   These 
relations over finite fields play an important role in quantum mechanics as they model stabilizer circuits for $p$-dits \cite{comfort2021graphical}.  They now become the main 
object of study in this paper.

\begin{definition}
A \textbf{symplectic form} is a bilinear map $\langle \_,\_\rangle: V \otimes V\to F$ which satisfies:
\begin{itemize}
\item $\langle x,x \rangle=0$ for every $x \in V$;
\item Any $x$ such that $\langle{x,y}\rangle=0$ for every $y \in V$ implies $x=0$.
\end{itemize}
A vector space $V$ equipped with a symplectic form $\langle\_,\_\rangle$ is called a \textbf{symplectic vector space.}
\end{definition}

One can always express a bilinear form in matrix form as $\langle x,y \rangle:= x^T  S y$.  For a symplectic form one always has  $S^{-1}=-S=S^{T}$.
Furthermore, a symplectic vector space always has a Darboux basis:  this is a basis $q_1,...,q_n,p_1,...,p_n$ for which 
$\langle q_i,q_j \rangle = \langle p_i,p_j \rangle = \delta_{i,j}$ and $\langle q_i,p_j \rangle = \langle p_i,q_j \rangle = -\delta_{i,j}$ (where $\delta_{i,j} =1$ when $i=j$ and 
zero otherwise).  In particular, this means that a symplectic vector space always has even dimension.  In the special case when the canonical basis of the vector space 
is a Darboux basis, the symplectic form can be expressed using the $2n \times 2n$ block matrix $J$: 
\[ \langle x,y \rangle:=  \begin{pmatrix} q^T & p^T \end{pmatrix} J \begin{pmatrix} q \\ p \end{pmatrix}  ~~~\mbox{where} ~~~
J =\begin{pmatrix}
0 & 1_{n}\\
-1_{n} & 0\\
\end{pmatrix}
\]

When a symplectic vector space has its canonical basis a Darboux basis, its vectors are naturally graded. 
In classical physics one regards $q^T=(q_1 \ldots, q_n)$ as position and $p^T=(p_1 \ldots p_n)$ as momenta coordinates.  

\begin{definition}
The \textbf{symplectic dual} of a linear subspace $U \subseteq V$ of a symplectic vector space $V$ is the linear subspace 
$U':=\{u' \in V | \forall u \in U, \langle u',u \rangle=0\} \subseteq V$.
\end{definition}

Suppose a linear subspace $U \subseteq V$ is specified by the generator matrix $G$ then a vector 
$u' \in V$ is in $U'$ if and only of it satisfies $u' S G t = 0$ for all $t \in F^k$. This means that $U'$ can be specified using 
the parity-check matrix $H' := G  S^T$, which means $H' S = G$.  Similarly, letting $H$ be a parity-check matrix  
matrix for $U$ then one obtains a generator matrix $G'$ for $U'$ as  $G':= H S$.  Thus, similarly to the situation for orthogonal 
subspaces, for symplectic duals one can flip between generating and parity-checking.

\begin{definition}
A \textbf{Lagrangian subspace} $U$ of a symplectic vector space $V$ is linear subspace which is its own symplectic dual, so that $U=U'$.
\end{definition} 

When a subspace $U \subseteq V$ of a symplectic vector space is Lagrangian one can obtain from a parity-check matrix $H$ a generator matrix $G$, using the above 
observations, as $G = H S$.  In fact, it is clear that a subspace is Lagrangian if and only if a parity-check and generator matrices can be chosen to satisfy
$G=H S$.  This has the consequence that the codomain of $G$ and $H$ must have the same dimension, so if $V$ has dimension $2n$ a Lagrangian relation 
must have dimension $n$.

We shall now specialize to the case in which the canonical basis and the Darboux basis coincide and take the states to be given by Lagrangian relations.  This gives a prop $\categoryname{LagRel}_F$ (see \cite{comfort2021graphical}) which, furthermore, is a hypergraph category: the Fr\"obenius structure is given by pairing black and white spiders.  Note that to obtain a relation from a state it is necessary to use the paired black and white cups and this introduces the negation in the last coordinate of Lagrangian relations which are described 
below. 

The prop $\categoryname{LagRel}_F$ may be described as follows:
\begin{description}
\item[Objects] $n \in \mathbb{N}$: correspond to the graded vector spaces $F^n \oplus F^n$ equipped with the cannonical symplectic form given by $J$, described above.  
\item[Maps]  $\mathcal{L}: n \to m$: are relations $\mathcal{R} \subseteq F^{n+m} \oplus F^{n+m}$  where the state $\mathcal{L} = \{ ((p,q),(p',-q'))|  (p,q),(p',q'))\in \mathcal{R} \}$ is 
Lagrangian.
\item[Tensor:]  Given by addition on objects and the (graded) direct sum on maps.
\end{description}

The composition of Lagrangian relations is relational composition and is described in \cite{fong2016algebra,comfort2021graphical}.

Our main interest is to obtain a characterization of the states of this prop -- which are of course just Lagrangian relations.  First observe that parity-check and generator matrices of Lagrangian relations can be graded
\[  H= 
\begin{pmatrix}
H_{q} ~|~ H_{p}\\
\end{pmatrix}  ~~~\mbox{and}~~~ G=
\begin{pmatrix}
G_{q} ~|~G_{p}\\
\end{pmatrix}  \]
so that, using the relation $G = H S$, discussed  above with $S:=J$, we have:
\begin{equation*}
\begin{pmatrix}
H_{q} ~|~ H_{p}\\
\end{pmatrix}
\begin{pmatrix}
0 & 1\\
-1 & 0\\
\end{pmatrix}
=\begin{pmatrix}
G_{q}\\
G_{p}\\
\end{pmatrix}
\end{equation*}
which implies that 
\begin{equation*}
    H_p=-G_q,\quad H_q=G_p.
\end{equation*}

Note that an arbitrary permutation of the coordinates view as a relation will not in general be a Lagrangian relation.  For a permutations to be a Lagrangian relation it must preserve the symplectic form and so be a {\em symplectomorphism\/} or a {\em symplectic permutation\/}.  Thus, it must not only respect the grading but also permute the two grades in the same manner.  This means a symplectic permutation can be written as:
\[ \sigma_S =  \begin{pmatrix} \sigma & 0 \\ 0 &\sigma \end{pmatrix}. \]

We saw that parity-check and generator matrices of a linear relation can be brought into a standard form. We can attempt to do the same for Lagrangian states, however, we will also want the permutation matrix $\sigma$ in \eqref{standard-formH} and \eqref{standard-formG} that permutes the columns of $G$ and $H$ to be Lagrangian. Taking this into account we can state the main result of the paper:

\begin{theorem}\label{standard form theorem} The parity-check matrix $H$ for a Lagrangian subspace $\mathcal R \subseteq (F^n)^2$ can be put into the following standard form:
\begin{equation}
    H = \begin{pmatrix} {Y} & 0 & 1_{n_p} & A^T\\
     -A & 1_{n_q} & 0 & 0  \end{pmatrix}\sigma_{S} \label{standard-form-lagrangian}
\end{equation}
where $n_p+n_q=n$ (the dimension of $\mathcal{R}$), $\sigma_S$ is a symplectic permutation, $A$ has dimensions $n_q \times n_p$, $Y$, and $Y=Y^T$. 
\end{theorem} 
\begin{proof}
First note that $H=(H_{q}|H_{p})$ and $HJH^{T}=0$ this implies that
\begin{equation*}
0=HJH^{T}=\begin{pmatrix}H_{q} | H_{p}\\
\end{pmatrix}
\begin{pmatrix}
0 & 1\\
-1 & 0\\
\end{pmatrix}
\begin{pmatrix}
H_{q}^{T}\\
H_{p}^{T}\\
\end{pmatrix}
=-H_{p}H_{q}^{T}+H_{q}H_{p}^{T} \implies H_{p}H_{q}^{T}=H_{q}H_{p}^{T}
\end{equation*}
This implies that $H_{p}H_{q}^{T}$ is self-transpose.
Consider $H=(H_{q}|H_{p})$ and putting $H_{q}$ in standard form that is:
\begin{equation*}
H_{q}=\begin{pmatrix}
1 & A^T \\
0 & 0\\
\end{pmatrix}\sigma_{q}
\end{equation*}
We can adjust $H_{p}$ to be $H_{p}=H'_{p}\sigma_{q}$ where $H'_{p}=\begin{pmatrix}
B_{1} & B_{2}\\
B_{3} & B_{4}\\
\end{pmatrix}$
Hence $H_{p}H_{q}^{T}$ is:
\begin{equation*}
H_{p}H_{q}^{T}=\begin{pmatrix}
B_{1} & B_{2}\\
B_{3} & B_{4}\\
\end{pmatrix}\sigma_{q}\sigma_{q}^{T}
\begin{pmatrix}
1 & 0\\
A & 0\\
\end{pmatrix}
\end{equation*}
Now since $\sigma_{q}^{T}=\sigma^{-1}_{q}$ we have the following:
\begin{equation*}
H_{p}H_{q}^{T}=
\begin{pmatrix}
B_{1}+B_{2}A & 0\\
B_{3}+B_{4}A & 0\\
\end{pmatrix}
\end{equation*}
Now since $H_{p}H_{q}^{T}$ is self-transpose, it follows that:
\begin{equation*}
B_{3}+B_{4}A=0
\end{equation*}
and
\begin{equation*}
B_{1}+B_{2}A=B_{1}^{T}+A^T B_{2}^{T}
\end{equation*}
But note that 
\begin{equation*}
(H_{q}|H_{p})= 
\begin{pmatrix}
1 & A & B_{1} & B_{2}\\
0 & 0 & B_{3} & B_{4}\\
\end{pmatrix}\sigma_{S}
\end{equation*}
has rank $n$ and so 
\begin{equation*}\textit{rank}
\begin{pmatrix}
0 & 0 & B_{2} & B_{4}\\
\end{pmatrix}=n_{p}
\end{equation*}
Note that $B_{2}$ depends on $B_{4}$ so $B_{4}$ has rank $n_{p}$ and so is invertible.   This means we can use row operations to reduce $B_{4}$ to the identity which results in the following form for $H_{p}$
\[
H_{p}=\begin{pmatrix}
B'_1 & B'_2 \\
B'_3 & B'_4\\
\end{pmatrix} := \begin{pmatrix}
Y & 0 \\
A' & 1\\
\end{pmatrix}
\]
where the same equations hold so $A'= -1 A=-A$ and $Y =  Y + 0 A = Y^T + A^T 0 = Y^T$.  Hence, finally, we get the desired form:
\begin{equation*}
    H = \begin{pmatrix} Y & 0 & 1_{n_p} & A^T\\
     -A & 1_{n_q} & 0 & 0  \end{pmatrix} \sigma_S 
\end{equation*}
where $\sigma_{S}$ is a symplectic permutation.
\end{proof}
Again, for any particular Lagrangian subspace this standard form, determined by the triple $(Y,A,\sigma)$, is not unique, however, as before, the $Y$ and $A$ depend on $\sigma$. Recall also this standard form is for states: for Lagrangian relations one must convert the relation to a state using Lagrangian cups.

As described in \cite{comfort2021graphical}, there is functor 
\[ L: \categoryname{LinRel_{F}}\to \categoryname{LagRel}_{F}: \mathcal{R} \mapsto \mathcal{R} \oplus \mathcal{R}^\perp\]
which pairs the relation with its orthogonal complement.  

Consider $L(\mathcal{R})$ where $\mathcal{R}$ is a state: to see that $L(\mathcal R)$ is a Lagrangian subspace, observe that one can choose the 
generator matrix of $L(\mathcal R)$ to be of the form
\begin{equation}
    G_{L(\mathcal R)}=\begin{pmatrix} G & 0 \\ 0 & -H \end{pmatrix}, \quad   H_{L(\mathcal R)}=\begin{pmatrix} H & 0 \\ 0 & G \end{pmatrix}
\end{equation}
where $G$ and $H$ are generator and parity-check matrices of $\mathcal R$. Then we observe that $G_{L(\mathcal R)} = H_{L(\mathcal R)} J$ 
verifying this is a (non-standard) parity-check generator pair for a Lagrangian relation.

Lemma \ref{ortho-lemma} tells us that the effect of $L$ on a linear function is :
\[L(\mathcal{R}_M) = \{ ((q,p),(q',p')) | q' = M q, p = M p' \}\] 


\subsection{Composition}

So far we have avoided discussion composition of linear relations (Lagrangian or otherwise) in terms of parity-check matrices.  The main reason for this 
is that describing the parity check matrix of a composite of two relations is a little more involved than might be expected.   Consequently reasoning about 
the composition is often easier using alternative characterizations of the relations (we will see this technique being employed heavily in the next section).   

That said composition can be calculated in terms of standardizing parity-check and generator matrices quite cleanly.  Curiously this does not seem to have been observed 
in the literature on parity-check and generator matrices so we have included a brief description of it here.

Composition of relations viewed categorically has two steps: first one forms the pullback of the two relations to be composed and, next, one has to quotient the pullback 
back to a relation.   Thus the first step is to form the pullback $P$ and the next is to quotient this to the composite relation ${\cal R}_2 \circ {\cal R}_1$:

 
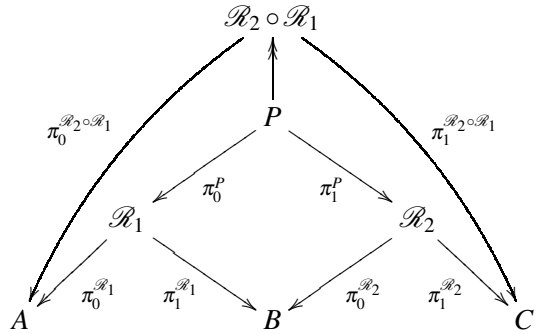
\[ \xymatrix{& & {\cal R}_2 \circ {\cal R}_1 \ar@/^1pc/[dddrr]^{\pi_1^{{\cal R}_2 \circ {\cal R}_1}} \ar@/_1pc/[dddll]_{\pi_0^{{\cal R}_2 \circ {\cal R}_1}} \\& & P \ar@{->>}[u]  \ar[dr]_{\pi_1^P} \ar[dl]^{\pi_0^P} \\ & {\cal R}_1 \ar[dl]^{\pi^{{\cal R}_1}_0}  \ar[dr]_{\pi^{{\cal R}_1}_1}&  & {\cal R}_2 \ar[dl]^{\pi^{{\cal R}_2}_0} \ar[dr]_{\pi^{{\cal R}_2}_1} \\  A & &  B & & C} \]
 
 This involves forming the pullback $P$ which is the same as forming the equalizer of 
 \[ \xymatrix{P \ar[rr]^{\langle \pi_0^P,  \pi_1^P \rangle} & & {\cal R}_1 \times {\cal R}_2 \ar@<1ex>[rr]^{\pi_0\pi_1^{{\cal R}_1}} \ar@<-1ex>[rr]_{\pi_1
 \pi_0^{{\cal R}_2}} & & B}. \]
 To calculate this one calculates a standardized parity-check matrix for  $ \pi_0\pi_1^{{\cal R}_1} - \pi_1\pi_0^{{\cal R}_2}$:
 \[ H  = \begin{pmatrix} 1  & A \end{pmatrix}: {\cal R}_1 \times {\cal R}_2 \to P. \] 
 One then ``inverts'' this to obtain a generator matrix $G^T = \begin{pmatrix} A \\ 1 \end{pmatrix}:  P \to {\cal R}_1 \times {\cal R}_2$.  Finally  one needs to calculate the 
 standardized generator matrix  for 
 \[ \langle \pi_0^P \pi_0^{{\cal R}_1}, \pi_1^P\pi_1^{{\cal R}_2} \rangle: P \to A \times C \]
 which is the map
 \[ \langle \pi_0^{{\cal R}_2 \circ {\cal R}_1}, \pi_1^{{\cal R}_2 \circ {\cal R}_1} \rangle: {\cal R}_2 \circ {\cal R}_1 \to A \times B. \]
 
 Here is an example to illustrate this algorithm:
  \begin{example}
 Consider linear maps $\mathcal{R}_{1}$ and $\mathcal{R}_{2}$ given as follows:
 \[
 \mathcal{R}_{1}=\begin{pmatrix}
 0 & 1\\
 1 & 0\\
 \end{pmatrix}~~~~~~~
 \mathcal{R}_{2}=\begin{pmatrix}
 0 & 1\\
 -1 & 0\\
 \end{pmatrix}
 \]
 The composition of these two linear maps is $\mathcal{R}_{3}$ given as follows:
 \[
 \mathcal{R}_{3}=
 \begin{pmatrix}
 0 & 1\\
 1 & 0\\
 \end{pmatrix}
 \begin{pmatrix}
 0 & 1\\
 -1 & 0\\
 \end{pmatrix}
 =
 \begin{pmatrix}
 -1 & 0\\
 0 & 1\\
 \end{pmatrix}
 \]
 To see this via composition of generator matrices first note that the generator matrix $G_{1}$ and $G_{2}$ corresponding to the relation $\mathcal{R}_{1}$ and $\mathcal{R}_{2}$ are:
 \[
 G_{1}=\begin{pmatrix}
 1 & 0 & 0 & 1\\
 0 & 1 & 1 & 0\\
 \end{pmatrix}~~~~~
 G_{2}=\begin{pmatrix}
 1 & 0 & 0 & 1\\
 0 & 1 & -1 & 0\\
 \end{pmatrix}
 \]
 The projection $\pi_{0}\pi_{1}^{\mathcal{R}_{1}}$ of the generator matrix $G_{1}$ is given as follows:
 \[
 \pi_{0}\pi_{1}^{\mathcal{R}_{1}}
 =
 \begin{pmatrix}
 1 & 0\\
 0 & 1\\
 0 & 0\\
 0 & 0\\
 \end{pmatrix}
 \begin{pmatrix}
 1 & 0 & 0 & 1\\
 0 & 1 & 1 & 0
 \end{pmatrix}
 \begin{pmatrix}
 0 & 0\\
 0 & 0\\
 1 & 0\\
 0 & 1\\
 \end{pmatrix}
 =\begin{pmatrix}
 0 & 1\\
 1 & 0\\
 0 & 0\\
 0 & 0\\
 \end{pmatrix}
 \]
 The projection $\pi_{1}\pi_{0}^{\mathcal{R}_{2}}$ for $G_{2}$ is given as follows:
 \[
 \pi_{1}\pi_{0}^{\mathcal{R}_{2}}
 =
 \begin{pmatrix}
 0 & 0\\
 0 & 0\\
 1 & 0\\
 0 & 1\\
 \end{pmatrix}
 \begin{pmatrix}
 1 & 0 & 0 & 1\\
 0 & 1 & -1 & 0
 \end{pmatrix}
 \begin{pmatrix}
 1 & 0\\
 0 & 1\\
 0 & 0\\
 0 & 0\\
 \end{pmatrix}
 =\begin{pmatrix}
 0 & 0\\
 0 & 0\\
 1 & 0\\
 0 & 1\\
 \end{pmatrix}
 \]
 Now one calculates the standardized parity-check using the formula  $ \pi_0\pi_1^{{\cal R}_1} - \pi_1\pi_0^{{\cal R}_2}$
 \[
  \pi_0\pi_1^{{\cal R}_1} - \pi_1\pi_0^{{\cal R}_2}=\begin{pmatrix}
  0 & 1\\
  1 & 0\\
  -1 & 0\\
  0 & -1\\
  \end{pmatrix}
 \]
 Now we calculate the $\pi^{p}_{0}\pi_{0}^{\mathcal{R}_{1}}$ and $\pi^{p}_{1}\pi_{1}^{\mathcal{R}_{2}}$
 \[
 \pi^{p}_{0}\pi_{0}^{\mathcal{R}_{1}}=
 \begin{pmatrix}
 0 & 1 & 1 & 0\\
 1 & 0 & 0 & 1\\
 \end{pmatrix}
 \begin{pmatrix}
 1 & 0\\
 0 & 1\\
 0 & 0\\
 0 & 0\\
 \end{pmatrix}
 \begin{pmatrix}
 1 & 0\\
 0 & 1\\
 \end{pmatrix}
 =\begin{pmatrix}
 0 & 1\\
 1 & 0\\
 \end{pmatrix}
 \]
 \[
 \pi^{p}_{1}\pi_{1}^{\mathcal{R}_{2}}=
  \begin{pmatrix}
 0 & 1 & 1 & 0\\
 1 & 0 & 0 & 1\\
 \end{pmatrix}
 \begin{pmatrix}
 0 & 0\\
 0 & 0\\
 1 & 0\\
 0 & 1\\
 \end{pmatrix}
 \begin{pmatrix}
 0 & 1\\
 -1 & 0\\
 \end{pmatrix}
 =\begin{pmatrix}
 0 & 1\\
 -1 & 0\\
 \end{pmatrix}
 \]
 Using this one can then write the composite generator matrix for $\mathcal{R}_{3}$ as:
 \[
 \begin{pmatrix}
 0 & 1 & 0 & 1\\
 1 & 0 & -1 & 0\\
 \end{pmatrix}
 \]
 Now putting this in standard form we get:
 \[
 G_{3}=\begin{pmatrix}
 1 & 0 & -1 & 0\\
 0 & 1 & 0 & 1 \\
 \end{pmatrix}
 \]
 which is the desired composite generator matrix for $R_{3}$
 \end{example}
 

\section{Kirchhoff Relations}


In this section we describe the category of Kirchhoff relations as a subcategory of $\categoryname{LagRel}$.  While similar categorical structures have been investigated in \cite{baez2015compositional, fong2016algebra}, we believe our approach, using parity-check matrices, besides being novel, brings some new subcategories to the fore.  


\subsection{Kirchhoff current law}


As remarked earlier objects in $\categoryname{LagRel}$ are graded vector spaces with the grading in classic physics being interpreted as position and momentum.   However, we wish now to switch interpretation so the grading represents voltage and current: to achieve this some additional structure seems necessary. The word ``current" implies the flow of some conserved quantity, while the word ``voltage'' implies a sort of gauge-invariance, which gives the freedom to shift all voltages by a constant.  Kirchhoff relations, which we now introduce, interpret these intuitions.  

\begin{definition}
A Lagrangian relation $\mathcal{R}: m \to n$  in $\categoryname{LagRel}$ satisfies: 
\begin{enumerate}[(i)]
\item \textbf{Kirchhoff's Current Law} if, for all $((q,p),(q',p')) \in \mathcal R$, the following equality holds:
\[  \sum_{j=1}^n p_j = \sum_{k=1}^m p'_k\]
\item \textbf{Translation invariance} if, whenever $\lambda \in F$  and $((q,p),(q',p')) \in \mathcal{R}$, then 
$((q+\vec{\lambda}_m,p),(q'+\vec{\lambda}_n,p')) \in \mathcal{R}$, where $\vec{\lambda}_n$ is a vector 
of dimension $n$ all of whose components are the same $\lambda \in F$.
\end{enumerate}
\end{definition} 

Kirchhoff's Current Law may be equivalently expressed as the requirement that $p^T \vec{1}_n  = {p'}^T \vec{1}_m$ while translation invariance amounts to 
insisting that $((\vec{1}_m,0),(\vec{1}_n,0))$ is a member of the relation. We now show that Kirchhoff's current law is, for Lagrangian relations,  equivalent to translation invariance . 
This result is essentially a version of Noether's theorem. 

\begin{proposition} \label{noether}
For a state $\mathcal{R}$ in $\categoryname{LagRel}_F$ the following are equivalent:
\begin{enumerate}[(i)]
    \item $\mathcal R$ satisfies the Kirchhoff current law;
    \item $\mathcal R$ satisfies translational invariance.
\end{enumerate}
\end{proposition}
\begin{proof}
Let $\mathcal R$ be specified by a generator matrix $G$ and a parity-check matrix $H$ which satisfy  $G=HJ$.
\begin{description}
\item{{\em (ii)} $\Rightarrow$ {\em (i)}:}
If $\mathcal{R}$ is translationally invariant then if $u \in \mathcal R$ this implies $(u+\lambda \epsilon)\in \mathcal{R}$ where 
$u:=\begin{pmatrix} q\\\ p \end{pmatrix}$ and $\epsilon_n:=\begin{pmatrix} \vec{1}_n \\ 0 \end{pmatrix}$.
This implies that $H(u + \lambda \epsilon_n) = H(u) + \lambda H(\epsilon)=0$ as $H(u)=0$ this means, in turn that $H(\epsilon_n) = 0$, and so $\epsilon_n \in \mathcal R$ by substituting $\lambda=1$. 
Now the symplectic product of $\epsilon_n$ with any $u \in \mathcal{R}$ must be zero so 
\[ 0 = \langle u, \epsilon_n \rangle = u^T J \epsilon = \begin{pmatrix} q^T & p^T \end{pmatrix} \begin{pmatrix} 0 & 1 \\ -1 & 0 \end{pmatrix} \begin{pmatrix} \vec{1}\\0 \end{pmatrix} = -p^T \vec{1} \]
implying that $p^T \vec{1}=0$ which is Kirchhoff current law.
\item{{\em (i)} $\Rightarrow$ {\em (ii):}}
Conversely suppose one has conservation of momentum then we have $0=p^T \vec{1}=u^T J \epsilon_n$ $\forall u\in \mathcal R$ however $u^T J \epsilon =0$ implies $\epsilon \in \mathcal{R}$ which in turn implies translational invariance, thereby concluding the proof.
\end{description}
\end{proof}

\begin{definition}
A Lagrangian relation satisfying either of the above (equivalent) conditions, translation invariance or the Kirchhoff current law, is a \textbf{Kirchhoff relation}.
\end{definition}

A linear relation $\mathcal L: m \to n$ in $\categoryname{LinRel}_R$ is a \textbf{quasi-stochastic} whenever it  relates $\vec{1}_m$ to $\vec{1}_n$.  Similarly, a matrix 
$A: m \to n$ in $\categoryname{Mat}_R$ is \textbf{quasi-stochastic} if $A \vec{1}_m = \vec{1}_n$.
Thus, for a quasi-stochastic relation, the most disordered input quasi-probabilistic distribution must be related to the most disordered output distribution. 
\begin{proposition}\label{quasi-lemma}
A Lagrangian relation is a Kirchhoff relation if and only if every standard parity-check matrix, as determined by $(Y,A,\sigma)$,
has $A$ quasi-stochastic and $Y \vec{1} = 0$.
\end{proposition}
\begin{proof}~
\begin{description}
\item{($\Rightarrow$)~}
As $\epsilon_n \in \mathcal{R}$ we have 
$0 = H_{\mathcal{R}} \epsilon_n = \begin{pmatrix} Y & 0 & 1 & A^T \\ -A & 1 & 0 & 0 \end{pmatrix}  \begin{pmatrix} \vec{1}_n \\ 0 \end{pmatrix} = \begin{pmatrix} Y \vec{1} \\ -A\vec{1}_{n_q} + \vec{1}_{n_p}  \end{pmatrix}.$
So $A^T \vec{1}_{n_q} = \vec{1}_{n_p}$ (thus $A^T$ is quasi-stochastic) and $Y \vec{1} = 0$.
\item{($\Leftarrow$)~} Conversely, if $A$ is quasi-stochastic and $Y \vec{1} = 0$ then  $H_{\mathcal{R}} \epsilon_n = 0$ and $\epsilon_n \in {\cal R}$ showing translation invariance.
\end{description}
\end{proof}
We now show that the category of Kirchhoff relations forms a subcategory of the prop of Lagrangian relations:
\begin{lemma}
The category of all Kirchhoff relations forms a prop which we denote \categoryname{KirRel}.
\end{lemma}
\begin{proof}
To show that the category of Kirchhoff relations forms a prop we need to prove they are closed to direct sums ($\_\oplus\_$) and composition.  

Let $\mathcal R_1: a \to b$ and let $\mathcal R_2: c \to d$ be Lagrangian relations which satisfy Kircchhoff's current law (KCL) then 
$\mathcal R_2 \oplus \mathcal R_1: a+c \to b+d$ also satisfies KCL, as, if 
\[((q_{a},p_{a}),(q_{b},p_{b}))\in \mathcal{R}_{1}~~~\mbox{and}~~~ ((q_{c},p_{c}),(q_{d},p_{d}))\in \mathcal{R}_{2}\]
 then $((q_a,q_c,p_a,p_c),(q_b,q_d,p_b,p_d))\in \mathcal{R}_{1}\oplus \mathcal R_{2}$. Then $\mathcal R_{1}\oplus \mathcal R_{2}$ satisfies KCL as:
\[ \begin{pmatrix} p_a^T &p_c^T \end{pmatrix} \vec{1}_{a+c} = p_a^T \vec{1} + p_c^T \vec{1} = p_b^T \vec{1} + p_d^T \vec{1} = \begin{pmatrix} p_b^T & p_d^T  \end{pmatrix} \vec{1}_{b+d}\]
where the middle step follows as the individual relations satisfy KCL.

To show closure under composition, let $\mathcal{R}_1: a \to b$ and $\mathcal{R}_2: b \to c$ be Lagrangian relations which satisfy KCL.  
$\mathcal{R}_2 \circ \mathcal{R}_1: a \to c$ satisfies KCL as, suppose $(({q}_a,{p}_a), ({q}_b, {p}_b)) \in \mathcal R_1$ and $(({q}_b,{p}_b),({q}_c, {p}_c)) \in \mathcal R_2$ we need to show that the composite $(({q}_a,{p}_a), ({q}_c, {p}_c)) \in \mathcal R_2 \circ \mathcal R_1$ obeys KCL. To show this first note that since $\mathcal R_{1}$ satisfies KCL we have $p_a^T \vec{1}= p_b^T \vec{1}$  and since $\mathcal R_{2}$ satisfies KCL we have $p_b^T \vec{1}={p}_c^T \vec{1}$ so that 
${p}_a^T \vec{1}  ={p}_c^T \vec{1}$ which shows that $\mathcal R_{2} \circ \mathcal R_{2}$ satisfies KCL.
\end{proof}


\subsection{Deterministic Relations }
\label{detkir}

We wish to consider Lagrangian relations which have a parity-check matrix determined by $(Y,A,\sigma)$ such that $A$ is not only quasi-stochastic but also is \textbf{deterministic}, in the sense that each row of $A$ has only one non-zero entry (which must be $1$).   Of course, this is not a well-defined notion unless, when one parity-check matrix has this form, all will have this form: 

\begin{lemma}
If a Lagrangian relation $\mathcal{R}$ has a standard parity-check matrix, $(Y,A,\sigma)$, with $A$ deterministic then every standard parity-check matrix of $\mathcal{R}$,  $(Y',A',\sigma')$ has $A'$ deterministic.
\end{lemma}

\begin{proof}
Consider the submatrix $\begin{pmatrix} -A  & 1_{n_{q}} \end{pmatrix}$ 
of the standard parity-check matrix which is deterministic.  This matrix has the property that each row has exactly one $1$ and one $-1$ to complete the proof we need to show that this property is invariant under the action of any permutation $\sigma$ which produces an alternative standard parity-check matrix.   The action of a $\sigma$ on this matrix maintains the property as $\sigma$ only permutes the columns. However, we need to show that when we put such a shuffled parity check-matrix back into standard form using Gaussian elimination that this property is still maintained.  To show this consider the following cases of Gaussian elimination:
\begin{description}
    \item{Permute rows:} In this case the property is maintained since changing the rows does not change the entries in the rows at all.
    \item{Negating a row:} This switches the $-1$ to a $1$ and a $1$ to a $-1$ and hence the property is still maintained.
    \item{Eliminating an entry:} This replaces a row by adding another rows to it with $1$, $-1$ in the same column. In this case $P$ is maintained only if the off-elimination 
    non-zero entries are in different columns.  However, notice if these off-elimination values are in the same column the rows are linearly dependent but this is not the case 
    since the matrix has full rank.
\end{description}
\end{proof}
This shows that Lagrangian relations which have a parity-check matrix $(Y,A,\sigma)$, in which $A$ is deterministic is a well-defined notion so we can make the definition:

\begin{definition}
A Lagrangian relation is \textbf{deterministic} in case any or all of its standard parity-check matrices, determined by $(Y,A,\sigma)$, has $A$ deterministic.
\end{definition}

Our objective, in this section is to explore these deterministic relations.  To do so we start with an alternative rather different description of them which allows us to develop their properties.

\begin{definition}
A Lagrangian relation $\mathcal{R}: n \to m$ is \textbf{position-partitioned} if there is an equivalence relation $\_\sim\_$ on $\underline{n+m} = \{ 1,...,n+m\}$ such that
\begin{itemize}
    \item For every $\begin{pmatrix} q \\ p \end{pmatrix} \in \mathcal{R}$ if $i \sim j$ then  $q_i = q_j$ (the relation is position-constant over equivalence classes);
    \item For each $i \in \underline{n+m}$, there is an $\begin{pmatrix} q \\ p \end{pmatrix} \in \mathcal{R}$ with $q_{j}=1$ when $i \sim j$  and $q_j = 0$ when $j \not\sim i$ (thus, the equivalence classes are position-separated).
\end{itemize}
\end{definition}

Observe that a Lagrangian relation cannot be position-partitioned in two different ways as, on the one hand, it must be position-constant over equivalence classes and,
on the other hand, position-separated on non-equivalent components.

Our first observation on position-paritioned relations is:

\begin{proposition}  
Position-partitioned relations form a subprop of $\categoryname{LagRel}$
\end{proposition}

\begin{proof}
We must show these relations are closed to composition and to direct sums.   For the latter, observe that when one takes a direct sum one takes the disjoint union of the entries: 
this means the partitions induced by the equivalence relation on each component induce a partition on the disjoint union.  This is clearly position-constant on partitions and can 
be separated by pairing the separating elements in each component with a zero in the other component.

For the composite $\mathcal{R}_2 \circ \mathcal{R}_1$  of two position-partitioned relations consider the composition of jointly epic cospans induced by the equivalence relations.   
Two entries are related, $i_1 \sim i_n$ if and only if there is a sequence $i_i \sim_{\delta_1} ... \sim_{\delta_{n-1}} i_n$ where $\delta_k = 1~\mbox{or}~2$.  If $\sim_1$ and $sim_2$ are position-constant then for each $\begin{pmatrix} q \\ p \end{pmatrix} \in \mathcal{R}_2 \circ  \mathcal{R}_1$ we have $q_j = q_{j+1}$  so $q_1 = q_n$ and $\mathcal{R}_2 \circ \mathcal{R}_1$ is position-constant over $\_\sim\_$.

Separation is given by taking the sum of the separators for all the partitions amalgamated by the composite relation.
\end{proof}

Using this definition one can prove the following proposition:

\begin{proposition}
A Lagrangian relation ${\cal R}$ is deterministic if and only if it is position-partitioned.
\end{proposition}

\begin{proof}~
\begin{description}
\item[$(\Rightarrow)$]  If one parity check matrix has $A$ deterministic then every parity check matrix has $A$ is deterministic. 
    Consider $H$ to be a parity check matrix for a deterministic Kirchhoff relation recall:
    \begin{equation*}
    \begin{pmatrix}
Y & 0 &  1 & A^T\\
-A & 1 & 0 & 0\\
\end{pmatrix}
\begin{pmatrix}
q_{0}\\
q_{1}\\
p_{0}\\
p_{1}\\
\end{pmatrix}
=
\begin{pmatrix}
Yq_{0}+p_{0}+A^{T}p\\
-Aq_{0}+q_{1}\\
\end{pmatrix}
= 0
\end{equation*}
which happens if and only if 
 $q_1 = Aq_{0}$ and $p_0 = -Yq_{0}-A^{T}p_{1}$.  But this means $q_0$ (and $p_1$) can be freely chosen.

As $A$ is deterministic, there is a surjective function $f:n_{q_{1}}\to n_{q_{0}}$ where $f:n_{q}\to n_{q_{0}}$  where $n_{q_{1}}+n_{q_{0}}=n$ such that $q_{i}=q_{f(i)}$ this defines an equivalence relation defined by $i\sim f(i)$ and $i \sim i'$ when $f(i)=f(i')$ for which the relation is position-constant.   The equation $Aq_{0}=q_{1}$ allows for separation of the equivalence, as one can then choose $p_{1}$ freely  to determine $p_{0}$.   Thus, we have position-separation as required.
\item[$(\Leftarrow)$] Assuming the relation is position-partitoned, and that we have a parity-check matrix for the relation $H=\begin{pmatrix}
Y & 0 & 1 & A^{T}\\
-A & 1 & 0 & 0\\
\end{pmatrix}$. Then, by assumption exists an equivalence relation $\_\sim\_$ (on $n$), for which  $\mathcal R$ is position-constant and separating.  Notice that $q_{0}$ can be arbitrary but this means each component of $q^{i}_{0}$ is in a separate equivalence class. The characteristic relation of the equivalence class determined by $q^{i}_{0}$ provides a column in $A$. The fact that we have an equivalence relation now guarantees each row has at most one $1$ and so $A$ is deterministic.
\end{description}
\end{proof}

In Lagrangian relations, whenever one has a position oriented notion, one expects a corresponding momentum oriented notion.  We describe this for states:

\begin{definition}
A Lagrangian state $\mathcal{R}: 0 \to n$ is \textbf{momentum-grouped} if there is an equivalence relation, $\_\sim\_$~, on $\underline{n}$ such that:
\begin{itemize}
    \item For every $S \subseteq \underline{n}$, which is $\sim$-closed (that is if $i \in S$ and $i \sim j$ then $j \in S$),  every $\begin{pmatrix} 0 \\ p \end{pmatrix} \in \mathcal{R}$  
    has $\sum_{i \in S} p_i =0$ (i.e. has zero average momentum at position zero on every set closed to the equivalence);
    \item For each $S \subseteq \underline{n}$ for which there are $i,j \in \underline{n+m}$ with $i \sim j$ and $i \in S$ and $j \not\in S$ there is a $\begin{pmatrix} 0 \\ p \end{pmatrix} \in \mathcal{R}$ such that $\sum_{j \in S} p_j \neq 0$ (i.e. has non-zero average momentum at position zero on every set which is not closed to the equivalence).
\end{itemize}
\end{definition}

\begin{proposition}
A Langrangian state is position-partitioned if and only if it is momentum-grouped.
\end{proposition}

\begin{proof}
It suffices show that for any parity-check matrix, determined by $(Y,A,\sigma)$, of a relation $\mathcal{R}$ which is momentum grouped has $A$ is deterministic.  The parity-check matrix 
implies the equation $Y q_0 + p_0 + A^T p_1 = 0$.  However,  we are interested in the case when $q_0 = 0$ so we are left with the equation $p_0 = - A^T p_1$.   Asking for $A$ to be deterministic is precisely to demand that $\mathcal{R}$ is momentum grouped. Note that $p_1$ can be arbitrarily chosen, when $A$ is deterministic composition with $A^T$ causes equivalent entries of $q_1$ to be summed to determine the entry of $p_0$ (which is an entry in the same equivalence class).  Notice that all the entries of $q_0$ are, therefore, in distinct  equivalence classes.  This guarantees the first condition of being momentum grouped.  The second condition is guaranteed as one may choose $q_1$ freely.
\end{proof}

This shows that the deterministic Lagrangian relations form a prop: we shall denote it by $\categoryname{DetLagRel}_{F}$.   It is now not hard to see that the Lagrangian relations satisfying the Kirchhoff current law  form a subprop, $\categoryname{KirRel}_F \subseteq \categoryname{LagRel}_F$.  Furthermore, deterministic Kirchhoff relations form a further subprop, $\categoryname{ResRel}_F \subseteq \categoryname{KirRel}_F$.  It is not hard to identified $\categoryname{ResRel}_F$ as precisely the prop of resistor circuits studied in \cite{fong2016algebra,baez2017props,baez2015compositional}.  Significantly, while Kirchhoff relations include resistor circuits, they also allow additional new components: namely ideal current dividers.


\subsection{Power input}
\label{powersec}


The power input, \cite{baez2015compositional, fong2016algebra,baez2017props}, of an electrical circuit is determined by the product of voltage (difference) times the current.  Similarly the product of momentum and velocity determines the energy of a physical system.   Power input is usually used to distinguish between {\em passive\/} and {\em active\/} components in electrical networks by their positive or negative power input requirements.  In the generalization to finite fields, the distinction between passive and active can no longer be made.   However, we can restrict to the {\em lossless\/} case when the power input is zero.   This give a series of subprops given by intersecting with lossless Lagrangian relations.
\begin{center}
\begin{tikzcd}
\sf{LagRel}_{F}                 &  & \sf{LossLagRel} \arrow[ll, hook]                       \\
\sf{KirRel}_{F} \arrow[u, hook] &  & \sf{LossKirRel} \arrow[ll, hook] \arrow[u, hook']      \\
\sf{ResRel}_{F} \arrow[u, hook] &  & \sf{Spiders} \arrow[u, hook] \arrow[ll, hook]
\end{tikzcd}
\end{center}
The purpose of the section is to brief explore these subprops.  

\begin{definition}
Given any Lagrangian relation, $\mathcal{R}: m \to n$, the associated \textbf{power input} is a function 
\[ P_{\mathcal R}:  \mathcal{R} \subseteq (F^m)^2 \otimes  (F^n)^2  \to F; ((q,p),(q',p')) \mapsto  \sum_{j=1}^n q_j p_j - \sum_{k=1}^m q'_k p'_k.\] 
A relation, $\mathcal{R}$ is said to be \textbf{lossless} in case $P_{\mathcal R}$ is everywhere zero.
\end{definition}

Spiders, symplectic permutations, and (ideal) current dividers are all lossless.  Notice that using spiders we can convert any relation to a state: because the cup changes the sign of output positions, a relation in $\categoryname{LagRel}$ has the same power input function as its corresponding state.

An explicit expression for the power input for the state $\mathcal R$ can be calculated as follows:
Consider the standard form for a parity-check matrix for a Lagrangian state:
\begin{equation*}
\begin{pmatrix}
Y & 0 & I & A^T\\
-A & I & 0 & 0\\
\end{pmatrix}
\begin{pmatrix}
q_{0}\\
q_{1}\\
p_{0}\\
p_{1}\\
\end{pmatrix}
=0
\end{equation*}
This gives the following relations:
\begin{equation}
q_{1}=A q_{0}~~~\mbox{and}~~~
p_{0}=-Y{q}_{0}-A^T{p}_{1} \label{conditionforll}
\end{equation}
The expression for power input using the above relations is then:
\[
   q_{0}^{T}p_{0}+q_{1}^{T}p_{1}= -q^{T}_{0}(Yq_{0}+A^Tp_{1})+(A q_{0})^{T}p_{1}=-q^{T}_{0}Yq_{0}
\]
For $\mathcal{R}$ to be lossless $q^{T}_{0}Yq_{0}$ must vanish. We now have the following lemma:
\begin{lemma}
In any field $F$ with $char(F) \neq 2$ when $Y=Y^T$ then  $q^T Y q = 0$ for every  $q$ if and only if $Y=0$.
\end{lemma}

\begin{proof}~
\begin{description}
\item{($\Rightarrow$)~}  Use the standard basis $e_i$ then $0 = e_i^T Y e_i = Y_{i,i}$ this shows that the diagonal entries of $Y$ are all zero. Now consider: 
\begin{equation*}
0 = (e_i+ e_j)^T Y (e_i+ e_j)=e_i^T Ye_i +e_i^T Y e_j  + e_j^T Ye_i +e_j^T Y e_j = Y_{i,i} + Y_{i,j}+Y_{j,i} + Y_{j,j} = Y_{i,j} + Y_{j,i}=2 Y_{i,j}
\end{equation*}
where the last two steps are  because the diagonal elements $Y_{i,i}$ and $Y_{j,j}$ are zero and then using the fact that $Y$ is self-transpose.  This gives $0=2 Y_{i,j}$. Since the characteristic of the field is not equal to $2$ we can divide by $2$ which implies that $Y_{i,j} = 0$.
\item{($\Leftarrow$)~} For the other direction we have If $Y = 0$ this implies that $q^{T}Yq=0$, for all $q$. This completes the proof.
\end{description}
\end{proof}
This implies that if $\mathcal R$ is lossless then $Y$ must be equal to zero, this gives us the following result:

\begin{proposition} Any state in \categoryname{LosLagRel} over $F$, can be described by a parity-check matrix in standard Lagrangian form with $Y=0$.
\end{proposition}

The kernel of a parity-check matrix with $Y=0$  
\begin{equation*}
0 = \begin{pmatrix} 0 & 0 & 1 & A^T \\-A & 1& 0 & 0 \end{pmatrix} \begin{pmatrix} q_0 \\ q_1 \\ p_0 \\ p_1 \end{pmatrix}  = \begin{pmatrix} p_0 + A^T p_1 \\ -A q_0 + q_1 \end{pmatrix}
\end{equation*} 
has an alternate description by the equatons $p_0=  -A^T  p_1$ and $q_1 =  A q_0$.  This gives:

\begin{proposition}
The category of Lossless Lagrangian Relations is isomorphic to the subcategory of linear relations determined by $L:\categoryname{LinRel}_F \to \categoryname{LagRel}_F$.
\end{proposition}\label{lossprop}
\begin{proof}
Consider the triangle shown below:
\begin{equation}
\begin{tikzcd}
\sf{LossLagRel} \arrow[rr, "\subset"] \arrow[rrd, "(2)", shift left=2] &  & \sf{LagRel}                                    \\
                                                                       &  & \sf{LinRel} \arrow[llu, "(1)"] \arrow[u, "L"']
\end{tikzcd}
\end{equation}
Recall that a functor $L:\categoryname{LinRel}\to \categoryname{LagRel}$ is faithful and bijective on objects.
We first show the direction $(1)$ by showing that $L(\mathcal L)$ is lossless. To show this note that because of (\ref{standard-formH}):
\begin{equation*}
\mathcal L=\{(q,q')~|~q=Aq'\}
\end{equation*}
Now by applying the functor $L$ to this relation we get:
\begin{equation*}
L(\mathcal L)=\{(q,p)~(q',p')~|~ q'=Aq,~p=A^{T}p'\}
\end{equation*}
but this is lossless as:
\begin{equation*}
q^{T}p-q'^{T}p'^{T}=q^{T}A^{T}p'-(Aq)^{T}p'=q^{T}A^{T}p'-q^{T}A^{T}p'=0
\end{equation*}
For direction $(2)$ one projects onto the $q$ coordinate.
\end{proof}

It follows that graphical linear algebra \cite{Bonchi2019GraphicalAA, pawelblog,zanasi,bonchi2017interacting} is a universal, sound and complete graphical calculus for \categoryname{LosLagRel}. 

\begin{corollary}
A lossless deterministic Kirchhoff relation is a spider.
\end{corollary}
\begin{proof}
If the relation is deterministic Kirchhoff this constrains $A$ to be deterministic and hence the relation contains no current dividers, furthermore the relation is lossless, this implies that no resistors are allowed. This only leaves \categoryname{Spiders} as elements of \categoryname{LosKirRel} if it is deterministic, thereby completeing the proof.
\end{proof}
\begin{remark}
From this result, it follows that \categoryname{GLA} forms a universal, sound and complete graphical calculus for \categoryname{LosLagRel}. This is a corresponds to a fragment of the graphical calculus for lagrangian relations, generated by $L(\text{copy spiders})$, $L(\text{addition spiders})$, and $L(\text{multiply by }k)$, with the same equations as those of GLA.
\end{remark}
\begin{remark}
The definition of power input makes sense for an arbitrary Lagrangian relation. So one can also define the categories of passive Lagrangian relations over $\mathbb R$, denoted as \categoryname{PasLagRel} which is defined as those Lagrangian relations $\mathcal R$ for which the power input $P_{\mathcal R}({q},{p})\geq 0$ for $({q},{p}) \in \mathcal R$. Similarly one can define the prop of \textbf{passive Kirchhoff relations} over $\mathbb R$, denoted as $\categoryname{PasKirRel(R)}$, are those relations $\mathcal R$ for which the power input $P_{\mathcal R}({q},{p})\geq 0$ for $({q},{p}) \in \mathcal R$. Physically, these correspond to elements that dissipate or conserve, but do not generate energy.

Resistors with negative resistances are not present in the prop of passive Kirchhoff relations over $\mathbb R$, but they are present in  Kirchhoff relations over $\mathbb R$.
\end{remark}

\subsection{Graph States}
\label{GS}
Although standard forms for both linear relations and  Lagrangian relations were defined, these standard forms are not unique. However, there is a special class of Kirchhoff relations, for which one can define a canonical form, and this will be useful in demonstrating universality. 

\begin{definition}
Any Kirchhoff relation $\mathcal R$ with $k_p=k$ (i.e ``no extra wires") is known as a {\bf graph state}.
\end{definition}

A graph state is specified by a parity-check matrix of the following form:
\[
{H} = \begin{pmatrix}
{Y}  & 1_{k\times k} 
\end{pmatrix}\sigma_{S}
\]
The following theorem will show that the $\sigma_{S}$ can be removed making the parity-check matrix form in the case of graph states unique:

\begin{theorem}\label{graphstatethm}
A graph state is uniquely specified by a parity-check matrix of the following canonical form:
\begin{equation*}
{H} = \begin{pmatrix}
{Y}  & 1_{k\times k} 
\end{pmatrix}
\end{equation*}
where the matrix ${Y}$ satisfies: ${Y}={Y}^T$ and ${Y} \vec{1}=0$.\label{thm:canonical}
\end{theorem}
\begin{proof}
First, observe that any symplectic permutation leaves $\begin{pmatrix} \vec{1} \\ 0 \end{pmatrix}$ invariant: $$\sigma_S \begin{pmatrix} \vec{1} \\ 0 \end{pmatrix} = \begin{pmatrix} \vec{1} \\ 0 \end{pmatrix} $$

Let us apply the constraint $${H}\begin{pmatrix} \vec{1} \\ 0 \end{pmatrix}=0,$$ to a parity-check matrix ${H}$ that is in standard form for a Lagrangian relation. (This is a necessary and sufficient condition for the relation defined by ${H}$ to be Kirchhoff.) By assumption, $k_q=0$. Then ${Y}\vec{1}=0$, and so one has:
\begin{equation*}
    {H} = \begin{pmatrix}
    {Y} & 1_{k\times k} 
    \end{pmatrix}\begin{pmatrix} \sigma & 0 \\  0 & \sigma \end{pmatrix}
    \label{kirchhoff-normal-form1}.
\end{equation*}

Note that one can choose $\sigma= 1_{k\times k}$. Absorbing $\sigma_S$ into the normal form for a Lagrangian relation one has the following has:
\begin{equation*}
    {H} = \begin{pmatrix}
    {Y}\sigma & \sigma. 
    \end{pmatrix}
\end{equation*}
By row operations alone, one can convert $\sigma$ into $ 1_{k\times k}$. This is equivalent to left-multiplying by $\sigma^{-1}=\sigma^T$. The final expression is: 
\begin{equation*}
    {H} = \begin{pmatrix}
    \sigma^T{Y}\sigma & 1_{k\times k} 
    \end{pmatrix}.
\end{equation*}
Note that $\tilde{{Y}} = \sigma^T{Y}\sigma$ satisfies the same conditions as ${Y}$. These conditions allow us to interpret the off-diagonal parts of ${Y}$ and $\tilde{{Y}}$ as adjacency matrices of weighted graphs, with weights in $F$. (The diagonal parts are determined by ${Y}\vec{1}=0$.) The weighted graphs described by ${Y}$ and $\tilde{{Y}}$ differ only by a relabeling of vertices. 
Hence the graph state is uniquely determined by the weighted graph whose adjacency matrix is the off-diagonal part of ${Y}$.
\end{proof}

Using that standard form for parity-check matrices in the equation ${H}{u}=0$, one has:
\begin{equation*}
{p}=-{Y} {q}
\end{equation*}
Thus any graph relation defines a function from the space of positions to momenta. Since it is a well-known fact that canonical forms exist for linear functions, this gives another explanation of Theorem \ref{thm:canonical}.

Electrical Engineers would call $Y$ an admittance matrix:  the condition that $Y = Y^T$ is exactly the condition required for an electrical circuit to satisfy {\em reciprocity}, \cite{sadiku}.   Notice that reciprocity holds more generally for an arbitrary Lagrangian relation.

\section{Universality for Kirchhoff Relations}

In this section, we give universal sets of generators for several of the subcategories discussed in the previous chapter, these universality proofs provide the interpretation of different subcategories of Kirchhoff relations as categories of electrical networks. The generators for Kirchhoff relations are built up in a series of steps. First the ``spiders" which give the basic hyper-graphical structure are discussed, to these resistors are added, and then it is shown that graph states can be obtained by a mesh of resistors. Finally current dividers are added to give a set of generators for an arbitrary map in $\categoryname{KirRel}$. Also considered are the generators of $\categoryname{LosKirRel}$ and $\categoryname{ResRel}$.
 
\subsection{A Universal Set of Generators for \categoryname{Spiders}}
The generators for the category of spiders \categoryname{Spider} are shown in Table \ref{Kirchhoff-spider-generator-table}. These are obtained from applying the functor $L:\categoryname{LinRel_{F}}\to \categoryname{LagRel_{F}}$ to the generators of \categoryname{LinRel} associated with the copy spider.

\begin{center}
\begin{table}
\begin{tabular}{|P{4cm}|P{3cm}|P{3.5cm}|P{3cm}|}
\hline
 Name & Generator & Definition in terms of \categoryname{GAA} & Relation: $p_i,~q_i \in F$ s.t.\\
\hline
\begin{gather*}
\text{Unit}
\end{gather*} &
\begin{center}
\begin{tikzpicture}[scale=0.4]
	\begin{pgfonlayer}{nodelayer}
		\node [style={red, very thick}] (0) at (-11, 4) {};
		\node [style=none] (1) at (-11, 6) {};
		\node [style=none] (2) at (-11, 6.75) {};
	\end{pgfonlayer}
	\begin{pgfonlayer}{edgelayer}
		\draw [style=boldedge] (1.center) to (0.center);
	\end{pgfonlayer}
\end{tikzpicture}
\end{center}
&
\begin{center}
\begin{tikzpicture}[scale=0.4]
	\begin{pgfonlayer}{nodelayer}
		\node [style=new style 0] (0) at (1.5, 4) {};
		\node [style=none] (1) at (1.5, 6) {};
		\node [style=white] (2) at (0.5, 4) {};
		\node [style=none] (3) at (0.5, 6) {};
		\node [style=none] (6) at (0.5, 6.75) {$p_{1}$};
		\node [style=none] (7) at (1.5, 6.75) {$q_{1}$};
	\end{pgfonlayer}
	\begin{pgfonlayer}{edgelayer}
		\draw (0.center) to (1);
		\draw (2.center) to (3);
	\end{pgfonlayer}
\end{tikzpicture}
\end{center}
&
\begin{gather*}
p_1=0
\end{gather*}
\\
\hline
\begin{gather*}
\text{Counit}
\end{gather*} & 
\begin{center}
\begin{tikzpicture}[scale=0.4]
	\begin{pgfonlayer}{nodelayer}
		\node [style={red, very thick}] (0) at (-9, 7) {};
		\node [style=none] (1) at (-9, 5) {};
		\node [style=none] (2) at (-9, 4.25) {};
		\node [style=none] (3) at (-9, 8) {};
	\end{pgfonlayer}
	\begin{pgfonlayer}{edgelayer}
		\draw [style=boldedge] (0) to (1.center);
	\end{pgfonlayer}
\end{tikzpicture}
\end{center}
&
\begin{center}
\begin{tikzpicture}[scale=0.4]
	\begin{pgfonlayer}{nodelayer}
		\node [style=new style 0] (0) at (1.5, 4.75) {};
		\node [style=none] (1) at (1.5, 2.75) {};
		\node [style=white] (2) at (0.5, 4.75) {};
		\node [style=none] (3) at (0.5, 2.75) {};
		\node [style=none] (4) at (0.5, 1.75) {$p_{1}$};
		\node [style=none] (5) at (1.5, 1.75) {$q_{1}$};
		\node [style=none] (6) at (0.5, 5.75) {};
		\node [style=none] (7) at (1.5, 5.75) {};
	\end{pgfonlayer}
	\begin{pgfonlayer}{edgelayer}
		\draw (1.center) to (0);
		\draw (3.center) to (2);
	\end{pgfonlayer}
\end{tikzpicture}
\end{center}
&
\begin{gather*}
p_1=0
\end{gather*}\\
\hline
\begin{gather*}
\text{Monoid}
\end{gather*}
&
\begin{center}
\begin{tikzpicture}[scale=0.4]
	\begin{pgfonlayer}{nodelayer}
		\node [style={red, very thick}] (0) at (1.5, 5) {};
		\node [style=none] (1) at (1.5, 6.5) {};
		\node [style=none] (2) at (0.5, 3.5) {};
		\node [style=none] (3) at (2.5, 3.5) {};
		\node [style=none] (4) at (0.5, 2.75) {};
		\node [style=none] (5) at (2.5, 2.75) {};
		\node [style=none] (6) at (1.5, 7.25) {};
	\end{pgfonlayer}
	\begin{pgfonlayer}{edgelayer}
		\draw [style=boldedge] (2.center) to (0.center);
		\draw [style=boldedge] (3.center) to (0.center);
		\draw [style=boldedge] (0.center) to (1.center);
	\end{pgfonlayer}
\end{tikzpicture}

\end{center}
&
\begin{center}
\begin{tikzpicture}[scale=0.4]
	\begin{pgfonlayer}{nodelayer}
		\node [style=white] (0) at (1.5, 5) {};
		\node [style=none] (1) at (0.5, 3.5) {};
		\node [style=none] (2) at (2.5, 3.5) {};
		\node [style=none] (3) at (1.5, 6.5) {};
		\node [style=none] (4) at (1.5, 3.5) {};
		\node [style=none] (5) at (3.5, 3.5) {};
		\node [style=new style 0] (6) at (2.5, 5) {};
		\node [style=none] (7) at (2.5, 6.5) {};
		\node [style=none] (8) at (0.5, 2.5) {$p_1$};
		\node [style=none] (9) at (1.5, 2.5) {$q_1$};
		\node [style=none] (10) at (2.5, 2.5) {$p_{2}$};
		\node [style=none] (11) at (3.5, 2.5) {$q_{2}$};
		\node [style=none] (12) at (1.5, 7.5) {$p_{3}$};
		\node [style=none] (13) at (2.5, 7.5) {$q_{3}$};
	\end{pgfonlayer}
	\begin{pgfonlayer}{edgelayer}
		\draw (3.center) to (0);
		\draw (0) to (1.center);
		\draw (0) to (2.center);
		\draw (7.center) to (6);
		\draw (6) to (4.center);
		\draw (6) to (5.center);
	\end{pgfonlayer}
\end{tikzpicture}

\end{center}
&
\begin{gather*}
p_3 = p_1+p_2 \\ q_3=q_1=q_2
\end{gather*}
\\
\hline
\begin{gather*}
\text{Comonoid}
\end{gather*}
&
\begin{center}
\begin{tikzpicture}[scale=0.4]
	\begin{pgfonlayer}{nodelayer}
		\node [style={red, very thick}] (0) at (1.5, -5) {};
		\node [style=none] (1) at (1.5, -6.5) {};
		\node [style=none] (2) at (0.5, -3.5) {};
		\node [style=none] (3) at (2.5, -3.5) {};
		\node [style=none] (4) at (0.5, -2.75) {};
		\node [style=none] (5) at (2.5, -2.75) {};
		\node [style=none] (6) at (1.5, -7.25) {};
	\end{pgfonlayer}
	\begin{pgfonlayer}{edgelayer}
		\draw [style=boldedge] (3.center) to (0);
		\draw [style=boldedge] (0) to (2.center);
		\draw [style=boldedge] (0) to (1.center);
	\end{pgfonlayer}
\end{tikzpicture}
\end{center}
&
\begin{center}
\begin{tikzpicture}[scale=0.4]
	\begin{pgfonlayer}{nodelayer}
		\node [style=white] (0) at (1.5, -5) {};
		\node [style=none] (1) at (0.5, -3.5) {};
		\node [style=none] (2) at (2.5, -3.5) {};
		\node [style=none] (3) at (1.5, -6.5) {};
		\node [style=none] (4) at (1.5, -3.5) {};
		\node [style=none] (5) at (3.5, -3.5) {};
		\node [style=new style 0] (6) at (2.5, -5) {};
		\node [style=none] (7) at (2.5, -6.5) {};
		\node [style=none] (8) at (0.5, -2.5) {$p_3$};
		\node [style=none] (9) at (1.5, -2.5) {$q_3$};
		\node [style=none] (10) at (2.5, -2.5) {$p_{3}$};
		\node [style=none] (11) at (3.5, -2.5) {$q_{3}$};
		\node [style=none] (12) at (1.5, -7.5) {$p_{1}$};
		\node [style=none] (13) at (2.5, -7.5) {$q_{1}$};
	\end{pgfonlayer}
	\begin{pgfonlayer}{edgelayer}
		\draw (3.center) to (0);
		\draw (0) to (1.center);
		\draw (0) to (2.center);
		\draw (7.center) to (6);
		\draw (6) to (4.center);
		\draw (6) to (5.center);
	\end{pgfonlayer}
\end{tikzpicture}
\end{center}
&
\begin{gather*}
p_1 = p_2+p_3 \\ q_3=q_1=q_2
\end{gather*}
\\
\hline
\end{tabular}
\caption{These are a universal set of generators for \categoryname{Spiders}. \label{Kirchhoff-spider-generator-table} 
}
\end{table}
\end{center}
Note the category of Lagrangian relations contains two inequivalent spiders, L(copy spider) and L(addition spider). Only one of these spiders is in the category of Kirchhoff relations, L(copy spider) -- which copies positions and adds momenta. We represent it by a red spider as shown in Table \ref{Kirchhoff-spider-generator-table}. That these generators are universal follows immediately from the fact that $L$ is a functor and copy spiders form a universal set of generators for the sub-prop \categoryname{copy} of \categoryname{LinRel} where the \categoryname{copy} is the category generated by the copy spiders in $\categoryname{LinRel}_{F}$.
One can use these spiders to define cups and caps for the category of Kirchhoff relations, this allows relation in Kirchhoff relations to  be converted into a state (or an effect). Borrowing terminology from electrical engineering, we sometimes refer to any relation $\mathcal R: F^{2n} \to F^{2m}$ as an $(n+m)$-terminal device.

\subsection{A Universal Set of Generators for \categoryname{KirRel}}
For Kirchhoff relations, we supplemented the generators in Tables \ref{Kirchhoff-spider-generator-table} with resistors shown in Table \ref{tab:resistor}, and ideal current dividers as shown in Table \ref{lossless-divider-table}.

\begin{table}
    \centering

\begin{tabular}{|P{2cm}|P{3.5cm}|P{3.5cm}|P{5cm}|}
\hline 
Name & Generator & Definition in terms of \categoryname{GAA} & Relation: $p_i,~q_i \in F$ s.t.\\
\hline
\begin{gather*}
\text{Resistor}
\end{gather*}
&
\begin{center}
\begin{tikzpicture}[scale=0.5]
	\begin{pgfonlayer}{nodelayer}
		\node [style={green,very thick}] (0) at (0, 3.75) {$y$};
		\node [style=none] (1) at (-1, 3.75) {};
		\node [style=none] (2) at (0, 6.25) {};
		\node [style=none] (3) at (0, 1) {};
		\node [style=none] (4) at (0, 0.25) {};
		\node [style=none] (5) at (0, 7.25) {};
	\end{pgfonlayer}
	\begin{pgfonlayer}{edgelayer}
		\draw [style=boldedge] (2.center) to (0);
		\draw [style=boldedge] (0) to (3.center);
	\end{pgfonlayer}
\end{tikzpicture}
\end{center}
&
\begin{center}
\begin{tikzpicture}[scale=0.5]
	\begin{pgfonlayer}{nodelayer}
		\node [style=white] (0) at (-3, 2.25) {};
		\node [style=scalarop] (1) at (-0.75, 3.25) {$y$};
		\node [style=new style 0] (2) at (-3, 2.25) {};
		\node [style=new style 4] (3) at (1.25, 4.25) {};
		\node [style=none] (4) at (-3, 6) {};
		\node [style=none] (5) at (-3, 0) {};
		\node [style=none] (6) at (1.25, 6) {};
		\node [style=none] (7) at (1.25, 0) {};
		\node [style=none] (8) at (-3, -1.25) {$p_{1}$};
		\node [style=none] (9) at (1.25, -1.25) {$q_{1}$};
		\node [style=none] (10) at (-3, 7) {$p_{2}$};
		\node [style=none] (11) at (1.25, 7) {$q_{2}$};
	\end{pgfonlayer}
	\begin{pgfonlayer}{edgelayer}
		\draw [in=30, out=-105] (1) to (0);
		\draw (6.center) to (3);
		\draw (3) to (7.center);
		\draw (2) to (5.center);
		\draw (2) to (4.center);
		\draw [in=-165, out=90] (1) to (3);
	\end{pgfonlayer}
\end{tikzpicture}
\end{center}
&
\begin{gather*}
p_{1}=p_{2}\\
p_{1}=y(q_{2}-q_1)\\
\end{gather*}\\ \hline
    \end{tabular}
    \caption{\textbf{ Resistor}: In addition to the generators for \categoryname{Spiders}, we must include \textit{resistors}, defined above, to obtain a universal set of generators for \categoryname{ResRel}.}
    \label{tab:resistor}
\end{table}
\begin{center}
\begin{table}
\begin{tabular}{|P{3.0cm}|P{8.0cm}|P{4.5cm}|}
\hline
Generator & Definition in terms of \categoryname{GAA} generators & Interpretation as a Kirchhoff Relation \\
\hline
\begin{center}
\begin{tikzpicture}[scale=0.4]
	\begin{pgfonlayer}{nodelayer}
		\node [style=white] (4) at (4, 4.25) {};
		\node [style=none] (5) at (4, 7.5) {};
		\node [style=none] (6) at (7, -0.5) {};
		\node [style=none] (7) at (0.5, -0.5) {};
		\node [style=green, very thick] (17) at (4, 4.25) {$~\mathbf{\langle{}\rangle_{w}}$};
	\end{pgfonlayer}
	\begin{pgfonlayer}{edgelayer}
		\draw [style=boldedge, bend left] (7.center) to (17);
		\draw [style=boldedge, bend left] (17) to (6.center);
		\draw [style=boldedge] (5.center) to (17);
	\end{pgfonlayer}
\end{tikzpicture}
\end{center}
&
\begin{center}
\begin{tikzpicture}[scale=0.5]
	\begin{pgfonlayer}{nodelayer}
		\node [style=white] (1) at (17.25, 5) {};
		\node [style=none] (2) at (17.25, 7.5) {};
		\node [style=none] (3) at (13, -1) {};
		\node [style=white] (4) at (22.5, 5.25) {};
		\node [style=none] (5) at (22.5, 7.5) {};
		\node [style=none] (6) at (26.5, -0.75) {};
		\node [style=none] (7) at (18.75, -1) {};
		\node [style=none] (8) at (13, -2.25) {$V_{1}$};
		\node [style=none] (9) at (18.75, -2.5) {$I_{1}$};
		\node [style=none] (10) at (21.5, -2.5) {$V_{2}$};
		\node [style=none] (11) at (26.75, -2.5) {$I_{2}$};
		\node [style=new style 0] (12) at (22.5, 5.25) {};
		\node [style=none] (13) at (17.25, 8.5) {$V_{3}$};
		\node [style=none] (14) at (22.5, 8.5) {$I_{3}$};
		\node [style=scalar] (15) at (13, 1.25) {$(1-w)$};
		\node [style=scalar] (16) at (18.75, 1.25) {$(1-w)^{-1}$};
		\node [style=scalar] (17) at (26.5, 1.25) {$w^{-1}$};
		\node [style=white] (18) at (17.25, 5) {};
		\node [style=none] (19) at (21.5, -1) {};
		\node [style=scalar] (20) at (21.5, 1.25) {$w$};
	\end{pgfonlayer}
	\begin{pgfonlayer}{edgelayer}
		\draw (1) to (2.center);
		\draw (4) to (5.center);
		\draw [bend right] (1) to (15);
		\draw (15) to (3.center);
		\draw (16) to (7.center);
		\draw [bend left=45, looseness=1.25] (16) to (12);
		\draw (17) to (6.center);
		\draw [bend left=45, looseness=1.25] (12) to (17);
		\draw [bend left] (18) to (20);
		\draw (20) to (19.center);
	\end{pgfonlayer}
\end{tikzpicture}
\end{center}
&
\begin{gather*}
V_{3}=V_{1}(1-w)+V_{2}w\\
I_{3}=(1-w)^{-1}I_{1}=w^{-1}I_{2}\\
\end{gather*}
\\
\hline
\begin{center}
\begin{tikzpicture}[scale=0.4]
	\begin{pgfonlayer}{nodelayer}
		\node [style=white] (4) at (0.5, 2.25) {};
		\node [style=none] (5) at (0.5, -1) {};
		\node [style=none] (6) at (3.5, 7) {};
		\node [style=none] (7) at (-3, 7) {};
		\node [style=green, very thick] (17) at (0.5, 2.25) {$~\mathbf{\langle{}\rangle^{w}}$};
	\end{pgfonlayer}
	\begin{pgfonlayer}{edgelayer}
		\draw [style=boldedge, bend right] (7.center) to (17);
		\draw [style=boldedge, bend right] (17) to (6.center);
		\draw [style=boldedge] (5.center) to (17);
	\end{pgfonlayer}
\end{tikzpicture}

\end{center}
&
\begin{center}
\begin{tikzpicture}[scale=0.5]
	\begin{pgfonlayer}{nodelayer}
		\node [style=white] (1) at (17.25, 1) {};
		\node [style=none] (2) at (17.25, -1.5) {};
		\node [style=none] (3) at (13, 7) {};
		\node [style=white] (4) at (22.5, 0.75) {};
		\node [style=none] (5) at (22.5, -1.5) {};
		\node [style=none] (6) at (26.5, 6.75) {};
		\node [style=none] (7) at (18.75, 7) {};
		\node [style=none] (8) at (13, 8.25) {$V_{1}$};
		\node [style=none] (9) at (18.75, 8.5) {$I_{1}$};
		\node [style=none] (10) at (21.5, 8.5) {$V_{2}$};
		\node [style=none] (11) at (26.75, 8.5) {$I_{2}$};
		\node [style=new style 0] (12) at (22.5, 0.75) {};
		\node [style=none] (13) at (17.25, -2.5) {$V_{3}$};
		\node [style=none] (14) at (22.5, -2.5) {$I_{3}$};
		\node [style=scalar] (15) at (13, 4.75) {$(1-w)^{-1}$};
		\node [style=scalar] (16) at (18.75, 4.75) {$(1-w)$};
		\node [style=scalar] (17) at (26.5, 4.75) {$w$};
		\node [style=white] (18) at (17.25, 1) {};
		\node [style=none] (19) at (21.5, 7) {};
		\node [style=scalar] (20) at (21.5, 4.75) {$w^{-1}$};
	\end{pgfonlayer}
	\begin{pgfonlayer}{edgelayer}
		\draw (1) to (2.center);
		\draw (4) to (5.center);
		\draw [bend left] (1) to (15);
		\draw (15) to (3.center);
		\draw (16) to (7.center);
		\draw [bend right=45, looseness=1.25] (16) to (12);
		\draw (17) to (6.center);
		\draw [bend right=45, looseness=1.25] (12) to (17);
		\draw [bend right] (18) to (20);
		\draw (20) to (19.center);
	\end{pgfonlayer}
\end{tikzpicture}
\end{center}
&
\begin{gather*}
V_{1}(1-w)^{-1}+V_{2}w^{-1}=V_{3}\\
I_{3}=I_{1}(1-w)=I_{2}w
\end{gather*}\\
\hline
\end{tabular}
\caption{\textbf{Ideal Current Dividers} A universal set of generators for \categoryname{LosKirRel} require, in addition to the spiders in Table \ref{Kirchhoff-spider-generator-table}, the above ``ideal current divider'' as defined in this table. \label{lossless-divider-table}
}
\end{table}
\end{center}
A resistor defines the following relationship between current and voltage:
\begin{equation*}
V_2-V_1=I_1 R,~I_1=I_2
\end{equation*}
This is sometimes known as the Ohm's law.

The power input of a resistor is given by:
\begin{equation*}
P=I^{2}R
\end{equation*}
A positive resistor (in the case $F=\mathbb{R}$) is usually considered to be a power dissipating (passive) electrical element, but if the resistor takes on negative values then it is an active element. Sometimes its more convenient to use \textbf{conductance} instead of resistance the conductance is just the inverse of resistance, that is:
\begin{equation*}
y=\frac{1}{R}.
\end{equation*}
Note that when the resistance across a wire is $0$ it is said to be ``short circuited", in that case the conductance is not defined, although we can informally regard it as $y=\infty$. In the case when conductance $y=0$ then the value of resistance is undefined, although we can be informally regard it as  $R=\infty$, and the relation is a disconnection. In electrical engineering this is said to be an ``open circuit": $I_1=I_2=0$, with $V_1$ and $V_2$ arbitrary.
\subsection{Ideal Current Dividers}
While the resistors and junctions of ideal wires are familiar from electrical engineering, the ideal current divider may seem unfamiliar. Let us discuss the physical interpretation of this generator.

One can view an ideal current divider as a limiting case of the following relation in Kirchhoff Relations over $\mathbb R$. Consider a three-terminal device, constructed from three resistors, $R_{12}=-\frac{1}{w} R$, $R_{13}=R$, and $R_{23}=\frac{1-w}{w} R$. Then, in the limit $R \to 0$, one can check that this describes the relation:
\begin{eqnarray}
I_1 & = & w I_3 \nonumber \\
I_2 & = & (1-w) I_3 \nonumber \\
V_3 & = & w V_1+(1-w) V_2 \nonumber
\end{eqnarray}
which is that of a ideal current divider. Thus in practice current dividers are constructed using resistors, note at least one of these resistances must be negative, so in practice, a ideal divider requires active components to be constructed; although the power input can be made arbitrarily close to zero. 

By virtue of the above construction, we could choose not to regard the ideal current divider as an independent generator for the theory of Kirchhoff relations over $\mathbb R$.  However, for Kirchhoff relations over an arbitrary, possibly-finite field, the above limiting procedure does not make sense, so we have no choice but to consider the ideal current divider as an independent generator.

\subsection{A Universal Set of Generators for \categoryname{LosKirRel}}
The category of $\categoryname{LosKirRel}$ generators include the generators of \categoryname{Spiders} in Table \ref{Kirchhoff-spider-generator-table}, supplemented with \textit{ideal current dividers} as shown in Table \ref{lossless-divider-table}.

\subsection{A universal set of generators for \categoryname{AffKirRel}}

To obtain the category of Affine Kirchhoff relations, we need two additional generators, shown in Table \ref{Kirchhoff-affine-generator-table}.

\begin{center}
\begin{table}
\begin{tabular}{|P{3cm}|P{3.5cm}|P{3.5cm}|P{3cm}|}
\hline
Name & Generator & Definition in terms of \categoryname{GAA} & Relation: $p_i,~q_i \in F$ s.t.\\
\hline
\begin{gather*}
\text{Voltage Source:}
\end{gather*} &
\begin{center}
\begin{tikzpicture}[scale=0.4]
	\begin{pgfonlayer}{nodelayer}
		\node [style={white, very thick}] (0) at (-7.5, -3.25) {$\mathbf \pm$};
		\node [style=none] (1) at (-7.5, 0) {};
		\node [style=none] (2) at (-7.5, -6.75) {};
		\node [style=none] (3) at (-9, -3.25) {$V$};
		\node [style=none] (17) at (-7.5, -8.25) {$1$};
		\node [style=none] (18) at (-7.5, 1.25) {$2$};

	\end{pgfonlayer}
	\begin{pgfonlayer}{edgelayer}
		\draw [style=boldedge] (1.center) to (0);
		\draw [style=boldedge] (0) to (2.center);
	\end{pgfonlayer}
\end{tikzpicture}
\end{center}
&
\begin{center}
\begin{tikzpicture}[scale=0.4]
	\begin{pgfonlayer}{nodelayer}
		\node [style=white] (0) at (3.25, -3) {};
		\node [style=none] (1) at (3.25, 0) {};
		\node [style=none] (2) at (3.25, -6.75) {};
		\node [style=none] (9) at (1.25, -6.75) {};
		\node [style=none] (10) at (4.25, 0) {};
		\node [style=none] (11) at (4.25, -7) {};
		\node [style=none] (12) at (3.25, 1.25) {$q_{2}$};
		\node [style=none] (13) at (4.25, 1.25) {$p_{2}$};
		\node [style=none] (14) at (3.25, -8.25) {$q_1$};
		\node [style=none] (15) at (4.25, -8.25) {$p_1$};
		\node [style=none] (16) at (0.75, -6.75) {};
		\node [style=none] (17) at (1.75, -6.75) {};
		\node [style=scalar] (18) at (1.25, -5) {$V$};
	\end{pgfonlayer}
	\begin{pgfonlayer}{edgelayer}
		\draw (1.center) to (0);
		\draw (0) to (2.center);
		\draw (10.center) to (11.center);
		\draw (16.center) to (17.center);
		\draw [bend right] (0) to (18);
		\draw (18) to (9.center);
	\end{pgfonlayer}
\end{tikzpicture}
\end{center}
&
\begin{gather*}
q_{2}=q_1+V\\
p_{2}=p_1\\
\end{gather*}
\\
\hline
\begin{gather*}
\text{Current Source:}
\end{gather*} & 
\begin{center}
\begin{tikzpicture}[scale=0.5]
	\begin{pgfonlayer}{nodelayer}
		\node [style={white, very thick}] (0) at (-7.5, -3.25) {$\mathbf \uparrow$};
		\node [style=none] (1) at (-7.5, 0) {};
		\node [style=none] (2) at (-7.5, -6.75) {};
		\node [style=none] (3) at (-9, -3.25) {$I$};
		\node [style=none] (17) at (-7.5, -8.25) {$1$};
		\node [style=none] (18) at (-7.5, 1.25) {$2$};
	\end{pgfonlayer}
	\begin{pgfonlayer}{edgelayer}
		\draw [style=boldedge] (1.center) to (0);
		\draw [style=boldedge] (0) to (2.center);
	\end{pgfonlayer}
\end{tikzpicture}

\end{center}
&
\begin{center}
\begin{tikzpicture}[scale=0.4]
	\begin{pgfonlayer}{nodelayer}
		\node [style=white] (0) at (1, -3) {};
		\node [style=none] (1) at (1, 0) {};
		\node [style=none] (2) at (1, -6.75) {};
		\node [style=none] (9) at (3, -6.75) {};
		\node [style=none] (10) at (-2.25, 0) {};
		\node [style=none] (11) at (-2.25, -7) {};
		\node [style=none] (12) at (1, 1.25) {$p_{2}$};
		\node [style=none] (13) at (-2.25, 1.25) {$q_{2}$};
		\node [style=none] (14) at (1, -8.25) {$p_1$};
		\node [style=none] (15) at (-2.25, -8.25) {$q_1$};
		\node [style=new style 0] (16) at (-2.25, -4) {};
		\node [style=new style 0] (17) at (-2.25, -2.5) {};
		\node [style=none] (18) at (2.5, -6.75) {};
		\node [style=none] (19) at (3.5, -6.75) {};
		\node [style=scalar] (20) at (3, -4.5) {$I$};
	\end{pgfonlayer}
	\begin{pgfonlayer}{edgelayer}
		\draw (1.center) to (0);
		\draw (0) to (2.center);
		\draw (10.center) to (17);
		\draw (16) to (11.center);
		\draw (18.center) to (9.center);
		\draw (9.center) to (19.center);
		\draw [bend left, looseness=1.25] (0) to (20);
		\draw (20) to (9.center);
	\end{pgfonlayer}
\end{tikzpicture}
\end{center}
&
\begin{gather*}
p_{2}=I+p_{1}
\end{gather*}\\
\hline
\end{tabular}
\caption{\textbf{Affine Kirchhoff Relations}: In addition to the generators of \categoryname{KirRel} we must add the voltage and current source generators to obtain a universal set of generators for \categoryname{AffKirRel} \label{Kirchhoff-affine-generator-table}}
\end{table}
\end{center}
The Tables \ref{Kirchhoff-spider-generator-table}, \ref{tab:resistor}, \ref{lossless-divider-table} and \ref{Kirchhoff-affine-generator-table}  give the generators for affine Kirchhoff relations $\categoryname{AffKirRel}$. This is discussed briefly in section \ref{aff}.
\subsection{Voltage Sources and Current Sources}
The current source ``forces'' the current to take on a certain value, while the voltage source, forces a specific voltage difference across two terminals. 

Any affine shift in position and momenta that obeys KCL can clearly be written in terms of these generators as shown in section \ref{aff}. Therefore, if one has a universal set of generators for \categoryname{KirRel}, then one can obtain any relation in \categoryname{AffKirRel}, by precomposing or post-composing with a direct sum of current sources, and a direct sum of voltage sources.

\subsection{Universality for Graph States}
We now argue that the resistor and the junction, allow one to obtain any graph state as defined in section \ref{GS}.
Observe that a resistor is not ``directional", in the sense that the diagrams of lemma \ref{reshor} are equal and so can be thought of as a weighted connection between two wires: we refer to the component as being ``horizontal". The following lemma shows that a resistor is horizontal in $\categoryname{KirRel}$ and discusses the matrix representation:

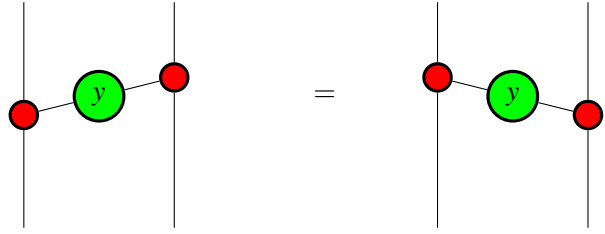
\begin{lemma}\label{reshor}
\[
\begin{tikzpicture}[scale=0.5]
	\begin{pgfonlayer}{nodelayer}
		\node [style=red, very thick] (0) at (-10, 4) {};
		\node [style=red, very thick] (1) at (-6, 5) {};
		\node [style=none] (2) at (-10, 7) {};
		\node [style=none] (3) at (-6, 7) {};
		\node [style=none] (4) at (-10, 1) {};
		\node [style=none] (5) at (-6, 1) {};
		\node [style=green, very thick] (6) at (-8, 4.5) {$y$};
		\node [style=none] (7) at (-8, 5.5) {};
		\node [style=none] (8) at (-10, -1) {};
		\node [style=none] (9) at (-6, -1) {};
		\node [style=none] (10) at (-10, 9) {};
		\node [style=none] (11) at (-6, 9) {};
		\node [style=none] (12) at (-2, 4.5) {$=$};
		\node [style=red, very thick] (13) at (1, 5) {};
		\node [style=red, very thick] (14) at (5, 4) {};
		\node [style=none] (15) at (1, 7) {};
		\node [style=none] (16) at (5, 7) {};
		\node [style=none] (17) at (1, 1) {};
		\node [style=none] (18) at (5, 1) {};
		\node [style=green, very thick] (19) at (3, 4.5) {$y$};
		\node [style=none] (20) at (3, 5.5) {};
		\node [style=none] (21) at (1, -1) {};
		\node [style=none] (22) at (5, -1) {};
		\node [style=none] (23) at (1, 9) {};
		\node [style=none] (24) at (5, 9) {};
	\end{pgfonlayer}
	\begin{pgfonlayer}{edgelayer}
		\draw (2.center) to (0);
		\draw (0) to (4.center);
		\draw (1) to (5.center);
		\draw (1) to (3.center);
		\draw (0) to (6);
		\draw (6) to (1);
		\draw (15.center) to (13);
		\draw (13) to (17.center);
		\draw (14) to (18.center);
		\draw (14) to (16.center);
		\draw (13) to (19);
		\draw (19) to (14);
	\end{pgfonlayer}
\end{tikzpicture}
\]
\label{horizontalres}
\end{lemma}
\begin{proof}The proof is by direct calculation, as shown below the constraints of the horizontal resistor are given as follows:
\[
I_{1}+yV_{1}-yV_{2}-I_{3}=0~~~-yV_{1}+I_{2}+yV_{2}=I_{4}
\]
\[
V_{1}=V_{3}~~~V_{2}=V_{4}
\]
This matrix can be written functionally as shown below, note that this relation is actually an isomorphism $M(y): (F^{2})^{2}\to (F^{2})^{2} $: 
\begin{equation*}
\begin{pmatrix}
1 & 0 & 0 & 0\\
0 & 1 & 0 & 0\\
y & -y & 1 & 0\\
-y & y & 0 & 1\\
\end{pmatrix}
\begin{pmatrix}
V_{1}\\
V_{1}\\
I_{1}\\
I_{2}\\
\end{pmatrix}
=
\begin{pmatrix}
V_{3}\\
V_{4}\\
I_{3}\\
I_{4}\\
\end{pmatrix}
\end{equation*}
\end{proof}
The concept of a ``mesh of resistors" is defined as follows:
\begin{definition}
A \textbf{mesh of resistors} is defined to be a weighted non-reflexive undirected graph where the weights are given by the values of the conductance $y\in F$ between pair of nodes. Each node is connected by convention to a unique output. This makes a mesh a state.
\end{definition}
An example for a mesh with $3$ conductance is given below:
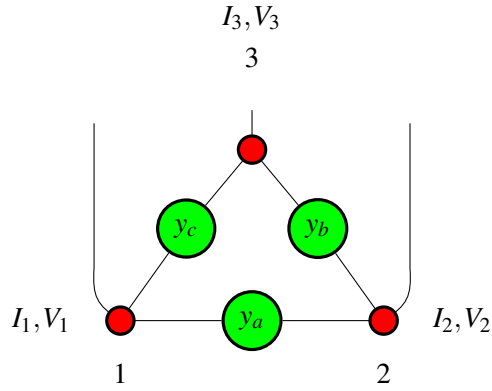
\begin{figure}
\centering
\begin{tikzpicture}[scale=0.7]
	\begin{pgfonlayer}{nodelayer}
		\node [style=none] (0) at (2, 0.25) {};
		\node [style=none] (1) at (7, 0.25) {};
		\node [style=none] (2) at (4.5, 3.5) {};
		\node [style=green, very thick] (3) at (3.25, 2) {$y_c$};
		\node [style=green, very thick] (4) at (5.75, 2) {$y_b$};
		\node [style=green, very thick] (5) at (4.5, 0.25) {$y_a$};
		\node [style=none] (6) at (1.5, 1.25) {};
		\node [style=none] (7) at (7.5, 1.25) {};
		\node [style=none] (8) at (4.5, 4.25) {};
		\node [style=none] (9) at (0.5, 0.25) {$I_{1}, V_{1}$};
		\node [style=none] (10) at (8.5, 0.25) {$I_{2}, V_{2}$};
		\node [style=none] (11) at (4.5, 6) {$I_{3}, V_{3}$};
		\node [style=none] (12) at (2, -0.75) {$1$};
		\node [style=none] (13) at (7, -0.75) {$2$};
		\node [style=none] (14) at (4.5, 5.25) {$3$};
		\node [style=red, very thick] (15) at (2, 0.25) {};
		\node [style=red, very thick] (16) at (7, 0.25) {};
		\node [style=red, very thick] (17) at (4.5, 3.5) {};
		\node [style=none] (18) at (1.5, 4) {};
		\node [style=none] (19) at (7.5, 4) {};
	\end{pgfonlayer}
	\begin{pgfonlayer}{edgelayer}
		\draw (0.center) to (5);
		\draw (5) to (1.center);
		\draw (2.center) to (4);
		\draw (4) to (1.center);
		\draw (0.center) to (3);
		\draw (3) to (2.center);
		\draw [bend right, looseness=1.25] (6.center) to (0.center);
		\draw [bend left, looseness=1.25] (7.center) to (1.center);
		\draw (2.center) to (8.center);
		\draw (18.center) to (6.center);
		\draw (19.center) to (7.center);
	\end{pgfonlayer}
\end{tikzpicture}
\caption{A mesh with $3$ resistors.}
\end{figure}

\begin{theorem} \label{Kirchhoff-graph-state-theorem}
A graph state with the following parity-check matrix:
\begin{equation*}
\begin{pmatrix}
{Y} & I\\
\end{pmatrix}
\begin{pmatrix}
{V}
_{in}\\
{I}_{in}\\
\end{pmatrix}
=
\begin{pmatrix}
0\\
0\\
\end{pmatrix}
\end{equation*}
can be realized by a mesh of resistors.
\end{theorem}
\begin{proof}
To prove this theorem consider a mesh of $n$ resistors.
\medskip
For each layer in the mesh we define a function $C_{ij}(y):F^{2n}\to F^{2n}$, and any graph state can be then drawn as a composition of resistors. For the case of $3$ resistors this is shown below:
\begin{figure}[ht]
    \centering
\begin{tikzpicture}[scale=0.5]
	\begin{pgfonlayer}{nodelayer}
		\node [style=red, very thick] (0) at (-4, -5.5) {};
		\node [style=red, very thick] (1) at (0, -5.5) {};
		\node [style=red, very thick] (2) at (4, -5.5) {};
		\node [style=none] (3) at (-4, 5.5) {};
		\node [style=none] (4) at (0, 5.5) {};
		\node [style=none] (5) at (4, 5.5) {};
		\node [style=none] (6) at (-4, 2.75) {};
		\node [style=none] (7) at (0, 2.75) {};
		\node [style=none] (8) at (0, 0.5) {};
		\node [style=none] (9) at (4, 0.5) {};
		\node [style=none] (10) at (-4, -1.75) {};
		\node [style=none] (11) at (4, -1.75) {};
		\node [style=green, very thick] (12) at (-0.75, -1.75) {$y_{c}$};
		\node [style=none] (13) at (0, -2.5) {};
		\node [style=none] (14) at (0, -2.5) {};
		\node [style=green, very thick] (15) at (-2, 2.75) {$y_{a}$};
		\node [style=green, very thick] (16) at (2, 0.5) {$y_{b}$};
		\node [style=none] (17) at (-4, 7) {$1$};
		\node [style=none] (18) at (0, 7) {$2$};
		\node [style=none] (19) at (4, 7) {$3$};
		\node [style=none] (20) at (-6, 2) {};
		\node [style=none] (21) at (-6, 3.5) {};
		\node [style=none] (22) at (5, 3.5) {};
		\node [style=none] (23) at (5, 2) {};
		\node [style=none] (24) at (-6, -0.25) {};
		\node [style=none] (25) at (-6, 1) {};
		\node [style=none] (26) at (5, -0.25) {};
		\node [style=none] (27) at (5, 1) {};
		\node [style=none] (28) at (-6, -2.5) {};
		\node [style=none] (29) at (-6, -1) {};
		\node [style=none] (30) at (5, -2.5) {};
		\node [style=none] (31) at (5, -1) {};
		\node [style=none] (32) at (6.25, -1.75) {$C_{13}(y_{a})$};
		\node [style=none] (33) at (6.25, 0.5) {$C_{23}(y_{b})$};
		\node [style=none] (34) at (6.25, 2.75) {$C_{12}(y_{c})$};
		\node [style=red] (35) at (-4, -1.75) {};
		\node [style=red] (36) at (4, -1.75) {};
		\node [style=red] (37) at (4, 0.5) {};
		\node [style=red] (38) at (0, 0.5) {};
		\node [style=red] (39) at (0, 2.75) {};
		\node [style=red] (40) at (-4, 2.75) {};
	\end{pgfonlayer}
	\begin{pgfonlayer}{edgelayer}
		\draw (0) to (3.center);
		\draw (2) to (5.center);
		\draw (1) to (13.center);
		\draw (14.center) to (8.center);
		\draw (8.center) to (7.center);
		\draw (7.center) to (4.center);
		\draw (10.center) to (12);
		\draw (12) to (11.center);
		\draw [style=new edge style 0] (8.center) to (16);
		\draw [style=new edge style 0] (16) to (9.center);
		\draw [style=new edge style 0, in=-180, out=0] (6.center) to (15);
		\draw [style=new edge style 0] (15) to (7.center);
		\draw [style=not edge] (20.center) to (21.center);
		\draw [style=not edge] (21.center) to (22.center);
		\draw [style=not edge] (20.center) to (23.center);
		\draw [style=not edge] (23.center) to (22.center);
		\draw [style=not edge] (24.center) to (26.center);
		\draw [style=not edge] (26.center) to (27.center);
		\draw [style=not edge] (24.center) to (25.center);
		\draw [style=not edge] (25.center) to (27.center);
		\draw [style=not edge] (28.center) to (30.center);
		\draw [style=not edge] (30.center) to (31.center);
		\draw [style=not edge] (28.center) to (29.center);
		\draw [style=not edge] (29.center) to (31.center);
	\end{pgfonlayer}
\end{tikzpicture}
    \caption{A graph state expressed as a composition of resistors}

\end{figure}
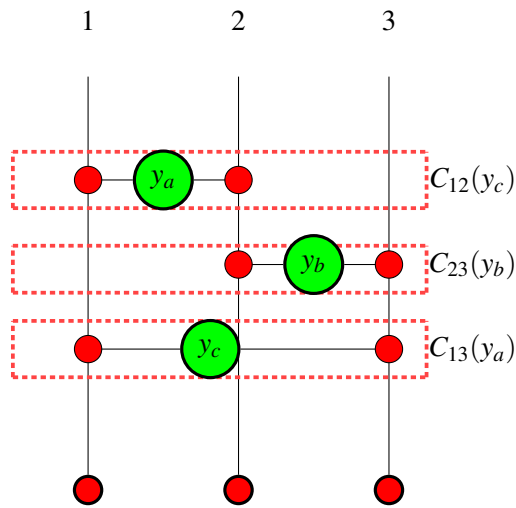

To define $C_{ij}(y)$ we first define the matrix $Y_{ij}(y)$ to be an $n\times{n}$ matrix whose entries are all zeroes except:
\[
(Y_{ij}(y))_{ii}=y~~(Y_{ij}(y))_{ij}=-y=(Y_{ij}(y))_{ji}~~(Y_{ij}(y))_{jj}=y
\]
Using this structure one can write down $C_{ij}(y)$ as follows:
\begin{equation*}
C_{ij}(y)=\begin{pmatrix}
1 & 0\\
Y_{ij}(y) & 1\\
\end{pmatrix}
\end{equation*}
The properties of $C_{ij}(y)$ are captured in the following lemma:
\begin{lemma}
The matrices $C_{ij}(y)$ satisfy the following properties:
\begin{enumerate}[(i)]
    \item $C_{ij}(y_{1})C_{kl}(y_{2})=C_{kl}(y_{2})C_{ij}(y_{1})$, these matrices commute.
    \item $Y_{ij}(y)$ is a symmetric matrix and $Y^{T}_{ij}(y)\vec{1}=0$, that is the rows sum to zero
    \item $\prod_{i\neq j}C_{ij}(y)=\begin{pmatrix}
     1 & 0\\
     \sum_{i\neq j} Y_{ij}(y) & 1\\
    \end{pmatrix}$
    \item $(\sum_{i\neq j}Y_{ij}(y))$ is symmetric and $(\sum_{i\neq j}Y_{ij}(y))^{T}\vec{1}=0$
\end{enumerate}
\end{lemma}
\begin{proof}~
\begin{description}
\item (i) To show that the $C$'s commutes one first writes the matrix structure explicitly as follows:
\[
C_{ij}(y_{1})=\begin{pmatrix}
 1 & 0\\
 Y_{ij}(y_{1}) & 1\\
\end{pmatrix}
~~~~~~
C_{kl}(y_{2})=\begin{pmatrix}
 1 & 0\\
 Y_{ij}(y_{2}) & 1\\
\end{pmatrix}
\]
The product $C_{ij}(y_{1})C_{kl}(y_{2})$ is: 
\[
C_{ij}(y_{1})C_{kl}(y_{2})=
\begin{pmatrix}
 1 & 0\\
 Y_{ij}(y_{1}) & 1\\
\end{pmatrix}
\begin{pmatrix}
 1 & 0\\
 Y_{ij}(y_{2}) & 1\\
\end{pmatrix}
=
\begin{pmatrix}
1 & 0\\
Y_{ij}(y_{1}) +Y_{kl}(y_{2})  & 1\\
\end{pmatrix}
= C_{kl}(y_{2})C_{ij}(y_{1})
\]
which completes the proof.
\item (ii) Recall, the structure of $Y_{ij}(y)$ is:
\[
(Y_{ij})(y))_{ii}=y~~(Y_{ij}(y))_{ij}=-y=(Y_{ij}(y))_{ji}~~(Y_{ij}(y))_{jj}=y
\]
and therefore its symmetric.
$Y_{ij}^{T}(y)\vec{1}$ sums the entries in each row and hence is always equal to zero given the structure defined above.
\item (iii) One can show this by direct calculation, this essentially relies on the fact that when lower triangular matrices of the following form are multiplied, the entries in the lower block add so we get the following result:
\[
\prod_{i\neq j}C_{ij}(y_{ij})=
\begin{pmatrix}
1 & 0\\
\sum_{i\neq j} Y_{ij}(y_{ij}) & 1\\
\end{pmatrix}
\]
\item (iv) $\sum_{i\neq j}Y_{ij}(y_{ij})$ is a symmetric matrix since its a sum of symmetric matrices. To show this $(\sum_{i\neq j}Y_{ij}(y))^{T}\vec{1}=0$, note that $(\sum_{i\neq j}Y_{ij}(y)))^{T}\vec{1}=\sum_{i\neq j}Y_{ij}^{T}(y)\vec{1}=0$, where the last equality follows from $(ii)$.
\end{description}
\end{proof}
The resistors for each layer are composed by multiplying the matrices $C_{ij}$, as shown in the lemma above $C_{ij}(y)$'s commute with each other and the product of $C_{ij}(y)$'s can then be written: 
\[
\prod_{i\neq j}C_{ij}(y_{ij})=
\begin{pmatrix}
1 & 0\\
\sum_{i\neq j} Y_{ij}(y_{ij}) & 1\\
\end{pmatrix}
\]
where $Y=\displaystyle{\sum_{i\neq j}} {Y}_{ij}(y_{ij})$ is a symmetric matrix, whose off-diagonal components are the adjacency matrix of a graph with weights defined by (minus) of the value of the conductance between wires $i$ and $j$, ${Y}_{ij}=-y_{ij}$, and whose diagonal components are defined via  ${Y}_{ii}=-\sum_{j\neq i}{Y}_{ij}(y_{ij})$.
The function encoded  the composition of the horizontal resistances, then has the following form:
\begin{equation*}
\begin{pmatrix}
1 & 0\\
{Y} & 1\\
\end{pmatrix}
\begin{pmatrix}
{V}_{in}\\
{I}_{in}\\
\end{pmatrix}
=
\begin{pmatrix}
{V}_{out}\\
{I}_{out}\\
\end{pmatrix}
\end{equation*}
The last layer consists of the effect, to ${V}_{out}$ and ${I}_{out}$.
The resulting relation allows each component of ${V}_{out}$ can be anything although each component of ${I}_{out}$ must be zero.
The parity-check matrix can then be written in the following form:
\begin{equation*}
\begin{pmatrix}
{Y} & 1\\
\end{pmatrix}
\begin{pmatrix}
{V}
_{in}\\
{I}_{in}\\
\end{pmatrix}
=
\begin{pmatrix}
0\\
0\\
\end{pmatrix}
\end{equation*}
Notice that every graph state can be obtained from a mesh of resistors by choosing appropriate conductance values and hence this completes the proof.
\end{proof}
To give an intuitive understanding of the steps used in the proof above, consider the following example:
\begin{example}
Consider a weighted graph, with weights in $F$. We convert this into a relation by associating each vertex of the graph with a spider, and each edge of the graph with a resistor, with non-zero resistance. The conductance associated with each resistance is the weight of the edge $y_{ij}$, which could be zero. An example, for a graph $3$ vertices is shown in Figure \ref{graphstate}.

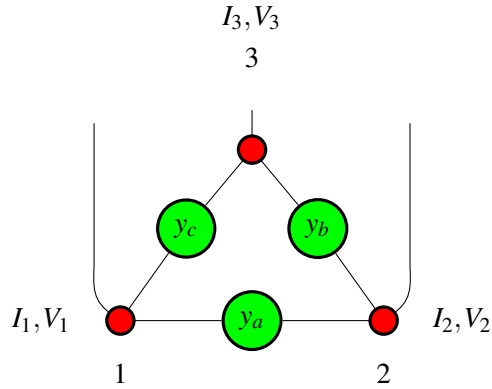
\begin{figure}
\centering
\begin{tikzpicture}[scale=0.7]
	\begin{pgfonlayer}{nodelayer}
		\node [style=none] (0) at (2, 0.25) {};
		\node [style=none] (1) at (7, 0.25) {};
		\node [style=none] (2) at (4.5, 3.5) {};
		\node [style=green, very thick] (3) at (3.25, 2) {$y_c$};
		\node [style=green, very thick] (4) at (5.75, 2) {$y_b$};
		\node [style=green, very thick] (5) at (4.5, 0.25) {$y_a$};
		\node [style=none] (6) at (1.5, 1.25) {};
		\node [style=none] (7) at (7.5, 1.25) {};
		\node [style=none] (8) at (4.5, 4.25) {};
		\node [style=none] (9) at (0.5, 0.25) {$I_{1}, V_{1}$};
		\node [style=none] (10) at (8.5, 0.25) {$I_{2}, V_{2}$};
		\node [style=none] (11) at (4.5, 6) {$I_{3}, V_{3}$};
		\node [style=none] (12) at (2, -0.75) {$1$};
		\node [style=none] (13) at (7, -0.75) {$2$};
		\node [style=none] (14) at (4.5, 5.25) {$3$};
		\node [style=red, very thick] (15) at (2, 0.25) {};
		\node [style=red, very thick] (16) at (7, 0.25) {};
		\node [style=red, very thick] (17) at (4.5, 3.5) {};
		\node [style=none] (18) at (1.5, 4) {};
		\node [style=none] (19) at (7.5, 4) {};
	\end{pgfonlayer}
	\begin{pgfonlayer}{edgelayer}
		\draw (0.center) to (5);
		\draw (5) to (1.center);
		\draw (2.center) to (4);
		\draw (4) to (1.center);
		\draw (0.center) to (3);
		\draw (3) to (2.center);
		\draw [bend right, looseness=1.25] (6.center) to (0.center);
		\draw [bend left, looseness=1.25] (7.center) to (1.center);
		\draw (2.center) to (8.center);
		\draw (18.center) to (6.center);
		\draw (19.center) to (7.center);
	\end{pgfonlayer}
\end{tikzpicture}
\caption{Graph State with 3 resistors. \label{graphstate}}
\end{figure}

This graph state, can be redrawn as a sequence of isomorphisms, each consisting of a single horizontal resistor acting on each pair of wires. For the example the Figure \ref{graphstate} can be redrawn as Figure \ref{horizontalgraphstate}.
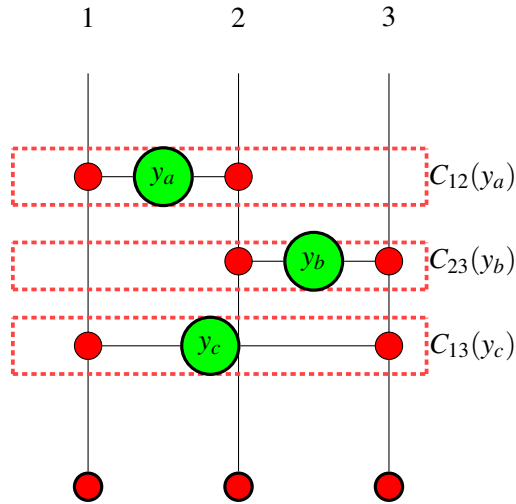
\begin{figure}[ht]
    \centering
\begin{tikzpicture}[scale=0.5]
	\begin{pgfonlayer}{nodelayer}
		\node [style=red, very thick] (0) at (-4, -5.5) {};
		\node [style=red, very thick] (1) at (0, -5.5) {};
		\node [style=red, very thick] (2) at (4, -5.5) {};
		\node [style=none] (3) at (-4, 5.5) {};
		\node [style=none] (4) at (0, 5.5) {};
		\node [style=none] (5) at (4, 5.5) {};
		\node [style=none] (6) at (-4, 2.75) {};
		\node [style=none] (7) at (0, 2.75) {};
		\node [style=none] (8) at (0, 0.5) {};
		\node [style=none] (9) at (4, 0.5) {};
		\node [style=none] (10) at (-4, -1.75) {};
		\node [style=none] (11) at (4, -1.75) {};
		\node [style=green, very thick] (12) at (-0.75, -1.75) {$y_{c}$};
		\node [style=none] (13) at (0, -2.5) {};
		\node [style=none] (14) at (0, -2.5) {};
		\node [style=green, very thick] (15) at (-2, 2.75) {$y_{a}$};
		\node [style=green, very thick] (16) at (2, 0.5) {$y_{b}$};
		\node [style=none] (17) at (-4, 7) {$1$};
		\node [style=none] (18) at (0, 7) {$2$};
		\node [style=none] (19) at (4, 7) {$3$};
		\node [style=none] (20) at (-6, 2) {};
		\node [style=none] (21) at (-6, 3.5) {};
		\node [style=none] (22) at (5, 3.5) {};
		\node [style=none] (23) at (5, 2) {};
		\node [style=none] (24) at (-6, -0.25) {};
		\node [style=none] (25) at (-6, 1) {};
		\node [style=none] (26) at (5, -0.25) {};
		\node [style=none] (27) at (5, 1) {};
		\node [style=none] (28) at (-6, -2.5) {};
		\node [style=none] (29) at (-6, -1) {};
		\node [style=none] (30) at (5, -2.5) {};
		\node [style=none] (31) at (5, -1) {};
		\node [style=none] (32) at (6.25, -1.75) {$C_{13}(y_{c})$};
		\node [style=none] (33) at (6.25, 0.5) {$C_{23}(y_{b})$};
		\node [style=none] (34) at (6.25, 2.75) {$C_{12}(y_{a})$};
		\node [style=red] (35) at (-4, -1.75) {};
		\node [style=red] (36) at (4, -1.75) {};
		\node [style=red] (37) at (4, 0.5) {};
		\node [style=red] (38) at (0, 0.5) {};
		\node [style=red] (39) at (0, 2.75) {};
		\node [style=red] (40) at (-4, 2.75) {};
	\end{pgfonlayer}
	\begin{pgfonlayer}{edgelayer}
		\draw (0) to (3.center);
		\draw (2) to (5.center);
		\draw (1) to (13.center);
		\draw (14.center) to (8.center);
		\draw (8.center) to (7.center);
		\draw (7.center) to (4.center);
		\draw (10.center) to (12);
		\draw (12) to (11.center);
		\draw [style=new edge style 0] (8.center) to (16);
		\draw [style=new edge style 0] (16) to (9.center);
		\draw [style=new edge style 0, in=-180, out=0] (6.center) to (15);
		\draw [style=new edge style 0] (15) to (7.center);
		\draw [style=not edge] (20.center) to (21.center);
		\draw [style=not edge] (21.center) to (22.center);
		\draw [style=not edge] (20.center) to (23.center);
		\draw [style=not edge] (23.center) to (22.center);
		\draw [style=not edge] (24.center) to (26.center);
		\draw [style=not edge] (26.center) to (27.center);
		\draw [style=not edge] (24.center) to (25.center);
		\draw [style=not edge] (25.center) to (27.center);
		\draw [style=not edge] (28.center) to (30.center);
		\draw [style=not edge] (30.center) to (31.center);
		\draw [style=not edge] (28.center) to (29.center);
		\draw [style=not edge] (29.center) to (31.center);
	\end{pgfonlayer}
\end{tikzpicture}
    \caption{The diagram of Figure \ref{graphstate} expressed as the composition of a series of resistors, followed by a direct sum of units.\label{horizontalgraphstate}}

\end{figure}

Each layer, highlighted in a red box in the example in Figure \ref{horizontalgraphstate}, defines a function $C_{ij}(y):F^{2n} \to F^{2n}$:
\begin{equation*}
    \begin{pmatrix} {V}' \\ {I}' \end{pmatrix} = 
    C_{ij}(y)   \begin{pmatrix} {V} \\ {I} \end{pmatrix}. 
\end{equation*}
where:
\begin{equation*}
C_{ij}(y)=
\begin{pmatrix}
1 & 0\\
{Y}_{ij}(y) & 1\\
\end{pmatrix}.
\end{equation*}
Explicitly, for the example in Figure \ref{graphstate}, $n=3$ and
\begin{equation*}
C_{12}=
\begin{pmatrix}
1 & 0 & 0 & 0 & 0 & 0\\
0 & 1 & 0 & 0 & 0 & 0\\
0 & 0 & 1 & 0 & 0 & 0\\
y_{a}& -y_{a} & 0 & 1 & 0 & 0\\
-y_{a} & y_{a} & 0 & 0 & 1 & 0\\
0 & 0 & 0 & 0 & 0 & 1\\
\end{pmatrix},
~~~~
C_{23}=
\begin{pmatrix}
1 & 0 & 0 & 0 & 0 & 0\\
0 & 1 & 0 & 0 & 0 & 0\\
0 & 0 & 1 & 0 & 0 & 0\\
0 & 0 & 0 & 1 & 0 & 0\\
0 & y_{b} & -y_{b} & 0 & 1 & 0\\
0 & -y_{b} & y_{b} & 0 & 0 & 1\\
\end{pmatrix},
\end{equation*}
and,
\begin{equation*}
C_{31}=
\begin{pmatrix}
1 & 0 & 0 & 0 & 0 & 0\\
0 & 1 & 0 & 0 & 0 & 0\\
0 & 0 & 1 & 0 & 0 & 0\\
y_{c} & 0 & -y_{c} & 1 & 0 & 0\\
0 & 0 & 0 & 0 & 1 & 0\\
-y_{c} & 0 & y_{c} & 0 & 0 & 1\\
\end{pmatrix}.
\end{equation*}
The product of the $C_{ij}$, can be written as 
\begin{equation*}
C_{12}C_{23}C_{31}  =
\begin{pmatrix}
1 & 0\\
{Y}_{12}(y_{A})+{Y}_{23}(y_{B})+{Y}_{31}(y_{C}) & 1\\
\end{pmatrix}= \begin{pmatrix}
1 & 0 & 0 & 0 & 0 & 0\\
0 & 1 & 0 & 0 & 0 & 0\\
0 & 0 & 1 & 0 & 0 & 0\\
y_{a}+y_{c} & -y_a & -y_{c} & 1 & 0 & 0\\
-y_a & y_a+y_b & -y_b & 0 & 1 & 0\\
-y_{c} & -y_b & y_b+y_{c} & 0 & 0 & 1\\
\end{pmatrix}
\end{equation*}
The above matrix can be written as a block matrix in the following way:
\begin{equation*}
\begin{pmatrix}
1 & 0\\
{Y} & 1\\
\end{pmatrix}
\begin{pmatrix}
{V}_{in}\\
{I}_{in}\\
\end{pmatrix}
=
\begin{pmatrix}
{V}_{out}\\
{I}_{out}\\
\end{pmatrix}
\end{equation*}
We now apply the unit to ${V}_{out}$ and ${I}_{out}$, as shown in the last layer of Figure \ref{horizontalgraphstate}, this results in the following parity-check matrix
\begin{equation*}
\begin{pmatrix}
{Y} & 1\\
\end{pmatrix}
\begin{pmatrix}
{V}
_{in}\\
{I}_{in}\\
\end{pmatrix}
=
\begin{pmatrix}
0\\
0\\
\end{pmatrix}
\end{equation*}
\end{example}

\subsection{Universality for Kirchhoff Relations}
In this section we wish to show that \textit{any} Kirchhoff relation can be specified as a composition or direct sum of the generators described in Tables \ref{Kirchhoff-spider-generator-table}, \ref{tab:resistor}, \ref{lossless-divider-table} and \ref{Kirchhoff-affine-generator-table}. Since any relation can be converted to a state using spiders, it suffices to show that any state can be expressed in terms of the generators listed above. Furthermore, the parity-check matrix of any Kirchhoff state can be written in standard form, so, to prove universality, we only need to provide an explicit construction of this standard form using our generators. 
\begin{theorem}\label{unilinkirrel}
Kirchhoff spiders, resistances and ideal current dividers form a universal set of generators for \categoryname{KirRel}.\end{theorem}
\begin{proof}
We first note that any Kirchhoff relation can be converted into a state using spiders. We showed that the parity-check matrix for any state can be put into the standard form described in the previous chapter. This depends on two matrices, the matrix ${Y}$ which satisfies ${Y}={Y}^T$ and ${Y}\vec{1}=\vec{1}$, and the matrix $A$ which is quasi-stochastic matrix. Our aim is to show one realizes a state with a parity-check matrix in this standard form using spiders, resistances and ideal current dividers.
Let us choose ${M}$ to be the following matrix:
\begin{equation*}
{M}=\begin{pmatrix}
{1}\\
{A}\\
\end{pmatrix}
\end{equation*}
where ${A}$ is the quasti-stochastic matrix appearing in the standard form of the state. Then ${M}$ is clearly also quasi-stochastic. The relation $L({M})$ can be realized using a layer of current-dividers composed with a layer of spiders.

Consider the composition of a graph state with $L(M)$ as shown in Figure \ref{fig:universality}. 
\begin{figure}[ht]
\centering
\begin{tikzpicture}[scale=0.5]
	\begin{pgfonlayer}{nodelayer}
		\node [style=none] (0) at (-4, 5) {};
		\node [style=none] (1) at (1, 5) {};
		\node [style=none] (2) at (-4, 3) {};
		\node [style=none] (3) at (1, 3) {};
		\node [style=none] (4) at (-1.5, 4) {$L({M})$};
		\node [style=none] (5) at (-3.5, 5) {};
		\node [style=none] (6) at (0.5, 5) {};
		\node [style=none] (7) at (-3.5, 3) {};
		\node [style=none] (8) at (0.5, 3) {};
		\node [style=none] (9) at (-3.5, 8) {};
		\node [style=none] (10) at (0.5, 8) {};
		\node [style=none] (11) at (0.5, 0) {};
		\node [style=none] (12) at (-3.5, 0) {};
		\node [style=none] (13) at (-6, -1) {};
		\node [style=none] (14) at (3, -1) {};
		\node [style=none] (15) at (-6, -3) {};
		\node [style=none] (16) at (3, -3) {};
		\node [style=none] (17) at (-3.5, -1) {};
		\node [style=none] (18) at (0.5, -1) {};
		\node [style=none] (19) at (-1.5, -2) {$\textit{Graph State}$};
		\node [style=none] (20) at (-1.5, 6) {$\ldots$};
		\node [style=none] (22) at (-1.5, 0) {$\ldots$};
		\node [style=none] (34) at (2.5, 0) {$({I}',{V}')$};
		\node [style=none] (36) at (2.5, 6) {$({I},{V})$};
	\end{pgfonlayer}
	\begin{pgfonlayer}{edgelayer}
		\draw (0.center) to (2.center);
		\draw (2.center) to (3.center);
		\draw (3.center) to (1.center);
		\draw (0.center) to (1.center);
		\draw (7.center) to (12.center);
		\draw (5.center) to (9.center);
		\draw (10.center) to (6.center);
		\draw (8.center) to (11.center);
		\draw (13.center) to (14.center);
		\draw (13.center) to (15.center);
		\draw (15.center) to (16.center);
		\draw (16.center) to (14.center);
		\draw (12.center) to (17.center);
		\draw (11.center) to (18.center);
	\end{pgfonlayer}
\end{tikzpicture}
    \caption{Any state in $\categoryname{KirRel}$ can be realized as a graph state, composed with a relation of the form $L(M)$ where $M$ is an appropriately chosen quasi-stochastic matrix.}
    \label{fig:universality}
\end{figure}
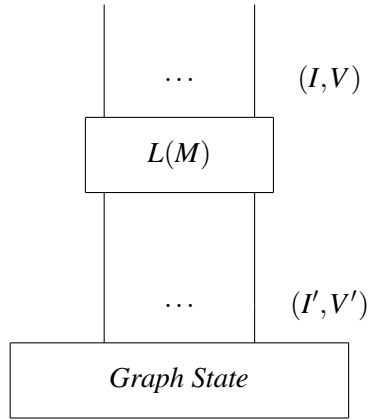
The relation encoded by the graph state is:
\begin{equation*}
\mathcal W=\{ (V',I')~|~
\begin{pmatrix}
{Y} & {1}\\
\end{pmatrix}
\begin{pmatrix}
{V}'\\
{I}'\\
\end{pmatrix}
=0
\}
=\{(V',I')~|~YV'=-I'\}
\end{equation*}
The relation $L(M)$ is the following:
\begin{equation*}
L(M)=\{((V,V')~(I,I'))~|~V=MV'~\textit{and}~I'=M^{T}I\}
\end{equation*}
We compose these relations as follows:
\begin{eqnarray*}
L(M)(1\oplus{\mathcal W})&=&\{(V,I)~|~\exists V',~I'~.~V=MV',~I'=M^{T}I,~YV'=-I'\}\\
&=&\{\begin{pmatrix}
V_{1}\\
V_{2}\\
\end{pmatrix},~\begin{pmatrix}
I_{1}\\
I_{2}\\
\end{pmatrix}~|~\exists~ V',~I'~.~ V_{1}=V',~V_{2}=AV',~I'=I_{1}+A^{T}I_{2},~YV'=-I'\}\\
&=&\{\begin{pmatrix}
V_{1}\\
V_{2}\\
\end{pmatrix},~\begin{pmatrix}
I_{1}\\
I_{2}\\
\end{pmatrix}~| V_{2}=AV_{1}, ~YV'=-I_{1}+A^{T}I_{2}\}\\
\end{eqnarray*}
Where the last step removes the existential quantification
We claim to have the following parity-check matrix for this relation:
\begin{equation*}
\begin{pmatrix}
{Y} & 0 & {M}^{T}\\
-{A} & {1} & 0\\
\end{pmatrix}
\end{equation*}
As for any $\left(\begin{pmatrix}
V_{1}\\
V_{2}\\
\end{pmatrix},~\begin{pmatrix}
I_{1}\\
I_{2}\\
\end{pmatrix}\right)\in L(M)(1\oplus{\mathcal W})$ we have:

\begin{equation*}
\begin{pmatrix}
{Y} & 0 & I & {A}^{T}\\
-{A} & {1} & 0 & 0\\
\end{pmatrix}
\begin{pmatrix}
V_{1}\\
V_{2}\\
I_{1}\\
I_{2}\\
\end{pmatrix}
=\begin{pmatrix}
YV_{1}+I_{1}-A^{T}I_{2}\\
-AV_{1}+V_{2}\\
\end{pmatrix}
=\begin{pmatrix}
0\\
0\\
\end{pmatrix}
\end{equation*}
This is the standard form for a parity-check matrix for a state in \categoryname{KirRel}. This concludes the proof.
\end{proof}

\begin{corollary}
If $A$ is deterministic as shown in section \ref{detkir} then only resistors and spiders are sufficient to generate all the maps in the category $\categoryname{ResRel}$.
\end{corollary}
\begin{example}
Consider the relation encoded by Figure \ref{extrawires}. This is an example of a state in \categoryname{ResRel} and we will now study how to write a parity-check matrix for such a relation:
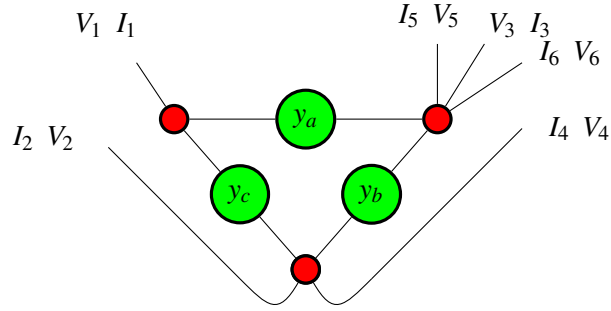
\begin{figure}
    \centering
\begin{tikzpicture}[scale=0.5]
	\begin{pgfonlayer}{nodelayer}
		\node [style=green, very thick] (0) at (-0.25, 3) {$y_a$};
		\node [style=red, very thick] (1) at (-3.75, 3) {};
		\node [style=red, very thick] (2) at (3.25, 3) {};
		\node [style=red, very thick] (3) at (-0.25, -1) {};
		\node [style=green, very thick] (4) at (-2, 1) {$y_c$};
		\node [style=green, very thick] (5) at (1.5, 1) {$y_b$};
		\node [style=none] (6) at (-4.75, 4.5) {};
		\node [style=none] (8) at (5.5, 4.5) {};
		\node [style=none] (11) at (-6, 5.5) {$V_{1}$};
		\node [style=none] (12) at (-5, 5.5) {$I_{1}$};
		\node [style=none] (13) at (5, 5.5) {$V_{3}$};
		\node [style=none] (14) at (6, 5.5) {$I_{3}$};
		\node [style=none] (17) at (-7.75, 2.5) {$I_{2}$};
		\node [style=none] (18) at (-6.75, 2.5) {$V_{2}$};
		\node [style=none] (19) at (3.5, 5.75) {$V_{5}$};
		\node [style=none] (20) at (2.5, 5.75) {$I_{5}$};
		\node [style=none] (25) at (3.25, 4) {};
		\node [style=none] (28) at (1.5, -1.25) {};
		\node [style=none] (29) at (3.25, 5) {};
		\node [style=none] (30) at (6.5, 2.75) {$I_{4}$};
		\node [style=none] (31) at (7.5, 2.75) {$V_{4}$};
		\node [style=none] (32) at (4.5, 5) {};
		\node [style=none] (33) at (7.25, 4.75) {$V_{6}$};
		\node [style=none] (34) at (6.25, 4.75) {$I_{6}$};
		\node [style=none] (37) at (5.5, 2.75) {};
		\node [style=red, very thick] (38) at (-0.25, -1) {};
		\node [style=none] (39) at (-2, -1.25) {};
		\node [style=none] (40) at (-5.5, 2.25) {};
	\end{pgfonlayer}
	\begin{pgfonlayer}{edgelayer}
		\draw (1) to (4);
		\draw (4) to (3);
		\draw (1) to (0);
		\draw (0) to (2);
		\draw (3) to (5);
		\draw (5) to (2);
		\draw (2) to (8.center);
		\draw (6.center) to (1);
		\draw [in=-135, out=-60, looseness=1.75] (3) to (28.center);
		\draw (2) to (29.center);
		\draw (2) to (32.center);
		\draw (28.center) to (37.center);
		\draw [in=-45, out=-120, looseness=1.75] (38) to (39.center);
		\draw (40.center) to (39.center);
	\end{pgfonlayer}
\end{tikzpicture}
    \caption{An example of a \textbf{state in \categoryname{ResRel}}. This state can be thought of as a graph state, composed with $L(M)$, where $M$ is a  deterministic matrix.
    \label{extrawires}}
\end{figure}
The relation is easily seen to be encoded by the following equations, relating $I_i$ and $V_i$ for $i=1$ to $6$:
\begin{equation*}
\begin{pmatrix}
Y & 1\\
\end{pmatrix}
\begin{pmatrix}
V_{1}\\
V_{2}\\
V_{3}\\
I_{1}\\
I_{2}+I_{4}\\
I_{3}+I_{5}+I_{6}\\
\end{pmatrix}
=\vec{0}
~~~V_{4}=V_{2}~~~V_{5}=V_{3}~~~V_{6}=V_{3}
\end{equation*}
We can encode these equations into the following parity-check matrix:
\begin{equation*}
\begin{pmatrix}
Y_{3\times 3} & 0_{3\times 3} & I_{3\times 3} & A_{3\times 3}\\
-A^{T}_{3 \times 3} & I_{3\times 3} & 0_{m\times{n}} & 0_{3\times{3}}\\
\end{pmatrix}
\begin{pmatrix}
{V}\\
{I}\\
\end{pmatrix}
\end{equation*}
where ${V}$ is the vector of $6$ voltages and ${I}$ is the vector of $6$ currents in the above figure \ref{extrawires}
where $A$ is the following matrix:
\begin{equation*}
A=
\begin{pmatrix}
0 & 0 & 0\\
1 & 0 & 0\\
0 & 1 & 1\\ 
\end{pmatrix}
~~~~
-A^{T}=
\begin{pmatrix}
0 & -1 & 0\\
0 & 0 & -1\\
0 & 0 & -1\\
\end{pmatrix}
\end{equation*}
$Y$ is adjacency matrix of the network shown in $\ref{extrawires}$.
\end{example}

\subsection{Kirchhoff Relations to Affine Kirchhoff Relations}\label{aff}
We first observe that any affine Kirchhoff transformation can be obtained from a linear Kirchhoff relation, by composing voltage and current sources to the inputs and outputs.
Assume without loss of generality that the linear relation is a state. Then we only need to show that any affine relation of the form
\begin{equation*} \forall \vec{q}, \vec{p} \in F^n, \quad
    \begin{pmatrix} {q} \\ {p} \end{pmatrix} \sim \begin{pmatrix} {q} \\ {p} \end{pmatrix} + \begin{pmatrix} {q}_0 \\ {p}_0 \end{pmatrix}, \quad \text{where }{q}_0,{p}_0 \in F^k,~\text{and }\vec{1}^T \cdot {p}_0 = 0
\end{equation*}
can be realized using our generators.

We first see that any shift in position can be realized using the voltage sources and any shift in currents that satisfies conservation of momentum can be obtained from current sources. This is true since the current and voltage sources obey KCL as shown in Table \ref{Kirchhoff-affine-generator-table}.

Notice that the current source can only realize shifts in momentum that satisfy conservation of momentum. Using these properties of current and voltage sources we have the following theorem:
\begin{theorem}\label{affkirrel}
Spiders, resistances, ideal current divders, voltage and current sources form a universal set of generators for $\categoryname{AffKirRel}$.
\end{theorem}
\begin{proof}
To extend the earlier proof for the universality of $\categoryname{KirRel}$ to $\categoryname{AffKirRel}$, we compose the diagram shown in figure \ref{fig:universality} with voltage and current sources to realize general affine shifts as shown in the figure below:
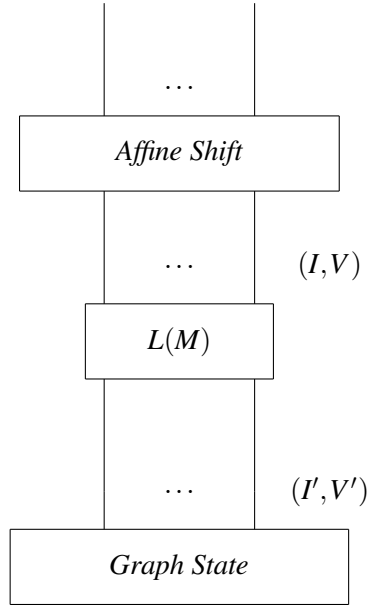
\begin{figure}
    \centering
\begin{tikzpicture}[scale=0.5]
	\begin{pgfonlayer}{nodelayer}
		\node [style=none] (0) at (-4, 5) {};
		\node [style=none] (1) at (1, 5) {};
		\node [style=none] (2) at (-4, 3) {};
		\node [style=none] (3) at (1, 3) {};
		\node [style=none] (4) at (-1.5, 4) {$L({M})$};
		\node [style=none] (5) at (-3.5, 5) {};
		\node [style=none] (6) at (0.5, 5) {};
		\node [style=none] (7) at (-3.5, 3) {};
		\node [style=none] (8) at (0.5, 3) {};
		\node [style=none] (9) at (-3.5, 8) {};
		\node [style=none] (10) at (0.5, 8) {};
		\node [style=none] (11) at (0.5, 0) {};
		\node [style=none] (12) at (-3.5, 0) {};
		\node [style=none] (13) at (-6, -1) {};
		\node [style=none] (14) at (3, -1) {};
		\node [style=none] (15) at (-6, -3) {};
		\node [style=none] (16) at (3, -3) {};
		\node [style=none] (17) at (-3.5, -1) {};
		\node [style=none] (18) at (0.5, -1) {};
		\node [style=none] (19) at (-1.5, -2) {$\textit{Graph State}$};
		\node [style=none] (20) at (-1.5, 6) {$\ldots$};
		\node [style=none] (22) at (-1.5, 0) {$\ldots$};
		\node [style=none] (34) at (2.5, 0) {$({I}',{V}')$};
		\node [style=none] (36) at (2.5, 6) {$({I},{V})$};
		\node [style=none] (37) at (-5.75, 10) {};
		\node [style=none] (38) at (2.75, 10) {};
		\node [style=none] (39) at (-5.75, 8) {};
		\node [style=none] (40) at (2.75, 8) {};
		\node [style=none] (41) at (-1.5, 9) {$\textit{Affine Shift}$};
		\node [style=none] (42) at (-3.5, 10) {};
		\node [style=none] (43) at (0.5, 10) {};
		\node [style=none] (44) at (-3.5, 8) {};
		\node [style=none] (45) at (0.5, 8) {};
		\node [style=none] (46) at (-3.5, 13) {};
		\node [style=none] (47) at (0.5, 13) {};
		\node [style=none] (48) at (-1.5, 10.75) {$\ldots$};
	\end{pgfonlayer}
	\begin{pgfonlayer}{edgelayer}
		\draw (0.center) to (2.center);
		\draw (2.center) to (3.center);
		\draw (3.center) to (1.center);
		\draw (0.center) to (1.center);
		\draw (7.center) to (12.center);
		\draw (5.center) to (9.center);
		\draw (10.center) to (6.center);
		\draw (8.center) to (11.center);
		\draw (13.center) to (14.center);
		\draw (13.center) to (15.center);
		\draw (15.center) to (16.center);
		\draw (16.center) to (14.center);
		\draw (12.center) to (17.center);
		\draw (11.center) to (18.center);
		\draw (37.center) to (39.center);
		\draw (39.center) to (40.center);
		\draw (40.center) to (38.center);
		\draw (37.center) to (38.center);
		\draw (46.center) to (42.center);
		\draw (47.center) to (43.center);
	\end{pgfonlayer}
\end{tikzpicture}
    \caption{Universality for \categoryname{AffKirRel}}
    \label{fig: Universality affine}
\end{figure}
Since any arbitrary affine shift in $\categoryname{AffKirRel}$ can be realized using a voltage and a current source, this completes the proof of universality for $\categoryname{AffKirRel}$.
\end{proof}

\section{Conclusion}
In this paper we have used parity check and generator matrices to study different subcategory of $\categoryname{LagRel}_{F}$. In particular we isolated a subcategory $\categoryname{KirRel}$ which obeyed Kirchoff's current law and so described electrical networks In doing this we establish connections between the structure of parity check matrices that are commonly used in quantum and classical error correction and structure of electrical networks. The various forms of the parity-check matrices which we investigated suggested a relation to normal forms for their corresponding electrical networks. It would be interesting to make this relation precise.
\section{Acknowledgements}
The authors would like to thank Cole Comfort for several illuminating discussions, comments and suggestions. We made use of open source tikz style files and the software 'Quantomatic' for drawing figures.
\bibliographystyle{eptcs}
\bibliography{parity-check-new}

\begin{thebibliography}{10}
\providecommand{\bibitemdeclare}[2]{}
\providecommand{\surnamestart}{}
\providecommand{\surnameend}{}
\providecommand{\urlprefix}{Available at }
\providecommand{\url}[1]{\texttt{#1}}
\providecommand{\href}[2]{\texttt{#2}}
\providecommand{\urlalt}[2]{\href{#1}{#2}}
\providecommand{\doi}[1]{doi:\urlalt{http://dx.doi.org/#1}{#1}}
\providecommand{\bibinfo}[2]{#2}

\bibitemdeclare{book}{sadiku}
\bibitem{sadiku}
\bibinfo{author}{Charles~K \surnamestart Alexander\surnameend} \&
  \bibinfo{author}{Matthew~NO \surnamestart Sadiku\surnameend}:
  \emph{\bibinfo{title}{Fundamentals of electric circuits}}.

\bibitemdeclare{article}{baez2017props}
\bibitem{baez2017props}
\bibinfo{author}{John~C \surnamestart Baez\surnameend},
  \bibinfo{author}{Brandon \surnamestart Coya\surnameend} \&
  \bibinfo{author}{Franciscus \surnamestart Rebro\surnameend}
  (\bibinfo{year}{2017}): \emph{\bibinfo{title}{Props in network theory}}.
\newblock {\sl \bibinfo{journal}{arXiv preprint arXiv:1707.08321}}.

\bibitemdeclare{article}{baez2014categories}
\bibitem{baez2014categories}
\bibinfo{author}{John~C \surnamestart Baez\surnameend} \&
  \bibinfo{author}{Jason \surnamestart Erbele\surnameend}
  (\bibinfo{year}{2014}): \emph{\bibinfo{title}{Categories in control}}.
\newblock {\sl \bibinfo{journal}{arXiv preprint arXiv:1405.6881}}.

\bibitemdeclare{article}{baez2015compositional}
\bibitem{baez2015compositional}
\bibinfo{author}{John~C \surnamestart Baez\surnameend} \&
  \bibinfo{author}{Brendan \surnamestart Fong\surnameend}
  (\bibinfo{year}{2015}): \emph{\bibinfo{title}{A compositional framework for
  passive linear networks}}.
\newblock {\sl \bibinfo{journal}{arXiv preprint arXiv:1504.05625}}.

\bibitemdeclare{article}{mitpaper}
\bibitem{mitpaper}
\bibinfo{author}{Mohsen \surnamestart Bahramgiri\surnameend} \&
  \bibinfo{author}{Salman \surnamestart Beigi\surnameend}
  (\bibinfo{year}{2006}): \emph{\bibinfo{title}{Graph states under the action
  of local Clifford group in non-binary case}}.
\newblock {\sl \bibinfo{journal}{arXiv preprint quant-ph/0610267}}.

\bibitemdeclare{article}{de2022circuit}
\bibitem{de2022circuit}
\bibinfo{author}{Niel \surnamestart de~Beaudrap\surnameend},
  \bibinfo{author}{Aleks \surnamestart Kissinger\surnameend} \&
  \bibinfo{author}{John \surnamestart van~de Wetering\surnameend}
  (\bibinfo{year}{2022}): \emph{\bibinfo{title}{Circuit Extraction for
  ZX-diagrams can be\# P-hard}}.
\newblock {\sl \bibinfo{journal}{arXiv preprint arXiv:2202.09194}}.

\bibitemdeclare{article}{boisseau2021string}
\bibitem{boisseau2021string}
\bibinfo{author}{Guillaume \surnamestart Boisseau\surnameend} \&
  \bibinfo{author}{Pawe{\l} \surnamestart Soboci{\'n}ski\surnameend}
  (\bibinfo{year}{2021}): \emph{\bibinfo{title}{String Diagrammatic Electrical
  Circuit Theory}}.
\newblock {\sl \bibinfo{journal}{arXiv preprint arXiv:2106.07763}}.

\bibitemdeclare{article}{gla}
\bibitem{gla}
\bibinfo{author}{Filippo \surnamestart Bonchi\surnameend},
  \bibinfo{author}{Joshua \surnamestart Holland\surnameend},
  \bibinfo{author}{Robin \surnamestart Piedeleu\surnameend},
  \bibinfo{author}{Pawe\l{} \surnamestart Soboci\'{n}ski\surnameend} \&
  \bibinfo{author}{Fabio \surnamestart Zanasi\surnameend}
  (\bibinfo{year}{2019}): \emph{\bibinfo{title}{Diagrammatic Algebra: From
  Linear to Concurrent Systems}}.
\newblock {\sl \bibinfo{journal}{Proc. ACM Program. Lang.}}
  \bibinfo{volume}{3}(\bibinfo{number}{POPL}), \doi{10.1145/3290338}.
\newblock \urlprefix\url{https://doi.org/10.1145/3290338}.

\bibitemdeclare{article}{Bonchi2019GraphicalAA}
\bibitem{Bonchi2019GraphicalAA}
\bibinfo{author}{Filippo \surnamestart Bonchi\surnameend},
  \bibinfo{author}{Robin \surnamestart Piedeleu\surnameend},
  \bibinfo{author}{Pawel \surnamestart Sobocinski\surnameend} \&
  \bibinfo{author}{Fabio \surnamestart Zanasi\surnameend}
  (\bibinfo{year}{2019}): \emph{\bibinfo{title}{Graphical Affine Algebra}}.
\newblock {\sl \bibinfo{journal}{2019 34th Annual ACM/IEEE Symposium on Logic
  in Computer Science (LICS)}}, pp. \bibinfo{pages}{1--12}.

\bibitemdeclare{inproceedings}{zanasi}
\bibitem{zanasi}
\bibinfo{author}{Filippo \surnamestart Bonchi\surnameend},
  \bibinfo{author}{Robin \surnamestart Piedeleu\surnameend},
  \bibinfo{author}{Pawel \surnamestart Sobociński\surnameend} \&
  \bibinfo{author}{Fabio \surnamestart Zanasi\surnameend}
  (\bibinfo{year}{2019}): \emph{\bibinfo{title}{Graphical Affine Algebra}}.
\newblock In: {\sl \bibinfo{booktitle}{2019 34th Annual ACM/IEEE Symposium on
  Logic in Computer Science (LICS)}}, pp. \bibinfo{pages}{1--12},
  \doi{10.1109/LICS.2019.8785877}.

\bibitemdeclare{article}{bonchi2017interacting}
\bibitem{bonchi2017interacting}
\bibinfo{author}{Filippo \surnamestart Bonchi\surnameend},
  \bibinfo{author}{Pawe{\l} \surnamestart Soboci{\'n}ski\surnameend} \&
  \bibinfo{author}{Fabio \surnamestart Zanasi\surnameend}
  (\bibinfo{year}{2017}): \emph{\bibinfo{title}{Interacting hopf algebras}}.
\newblock {\sl \bibinfo{journal}{Journal of Pure and Applied Algebra}}
  \bibinfo{volume}{221}(\bibinfo{number}{1}), pp. \bibinfo{pages}{144--184}.

\bibitemdeclare{article}{booth2021outcome}
\bibitem{booth2021outcome}
\bibinfo{author}{Robert~I \surnamestart Booth\surnameend},
  \bibinfo{author}{Aleks \surnamestart Kissinger\surnameend},
  \bibinfo{author}{Damian \surnamestart Markham\surnameend},
  \bibinfo{author}{Cl{\'e}ment \surnamestart Meignant\surnameend} \&
  \bibinfo{author}{Simon \surnamestart Perdrix\surnameend}
  (\bibinfo{year}{2021}): \emph{\bibinfo{title}{Outcome determinism in
  measurement-based quantum computation with qudits}}.
\newblock {\sl \bibinfo{journal}{arXiv preprint arXiv:2109.13810}}.

\bibitemdeclare{article}{Catani_2017}
\bibitem{Catani_2017}
\bibinfo{author}{Lorenzo \surnamestart Catani\surnameend} \&
  \bibinfo{author}{Dan~E \surnamestart Browne\surnameend}
  (\bibinfo{year}{2017}): \emph{\bibinfo{title}{Spekkens' toy model in all
  dimensions and its relationship with stabiliser quantum mechanics}}.
\newblock {\sl \bibinfo{journal}{New Journal of Physics}}
  \bibinfo{volume}{19}(\bibinfo{number}{7}), p. \bibinfo{pages}{073035},
  \doi{10.1088/1367-2630/aa781c}.
\newblock \urlprefix\url{https://doi.org/10.1088%2F1367-2630%2Faa781c}.

\bibitemdeclare{misc}{comfort2021graphical}
\bibitem{comfort2021graphical}
\bibinfo{author}{Cole \surnamestart Comfort\surnameend} \&
  \bibinfo{author}{Aleks \surnamestart Kissinger\surnameend}
  (\bibinfo{year}{2021}): \emph{\bibinfo{title}{A Graphical Calculus for
  Lagrangian Relations}}.

\bibitemdeclare{article}{fong2016algebra}
\bibitem{fong2016algebra}
\bibinfo{author}{Brendan \surnamestart Fong\surnameend} (\bibinfo{year}{2016}):
  \emph{\bibinfo{title}{The algebra of open and interconnected systems}}.
\newblock {\sl \bibinfo{journal}{arXiv preprint arXiv:1609.05382}}.

\bibitemdeclare{article}{fong2018seven}
\bibitem{fong2018seven}
\bibinfo{author}{Brendan \surnamestart Fong\surnameend} \&
  \bibinfo{author}{David~I \surnamestart Spivak\surnameend}
  (\bibinfo{year}{2018}): \emph{\bibinfo{title}{Seven sketches in
  compositionality: An invitation to applied category theory}}.
\newblock {\sl \bibinfo{journal}{arXiv preprint arXiv:1803.05316}}.

\bibitemdeclare{article}{fong2019hypergraph}
\bibitem{fong2019hypergraph}
\bibinfo{author}{Brendan \surnamestart Fong\surnameend} \&
  \bibinfo{author}{David~I \surnamestart Spivak\surnameend}
  (\bibinfo{year}{2019}): \emph{\bibinfo{title}{Hypergraph categories}}.
\newblock {\sl \bibinfo{journal}{Journal of Pure and Applied Algebra}}
  \bibinfo{volume}{223}(\bibinfo{number}{11}), pp. \bibinfo{pages}{4746--4777}.

\bibitemdeclare{article}{Gottesman1999}
\bibitem{Gottesman1999}
\bibinfo{author}{D.~\surnamestart {Gottesman}\surnameend}
  (\bibinfo{year}{1999}): \emph{\bibinfo{title}{{Fault-Tolerant Quantum
  Computation with Higher-Dimensional Systems}}}.
\newblock {\sl \bibinfo{journal}{Chaos Solitons and Fractals}}
  \bibinfo{volume}{10}, pp. \bibinfo{pages}{1749--1758},
  \doi{10.1016/S0960-0779(98)00218-5}.

\bibitemdeclare{article}{gross2006hudson}
\bibitem{gross2006hudson}
\bibinfo{author}{David \surnamestart Gross\surnameend} (\bibinfo{year}{2006}):
  \emph{\bibinfo{title}{Hudson’s theorem for finite-dimensional quantum
  systems}}.
\newblock {\sl \bibinfo{journal}{Journal of mathematical physics}}
  \bibinfo{volume}{47}(\bibinfo{number}{12}), p. \bibinfo{pages}{122107}.

\bibitemdeclare{book}{huffman2010fundamentals}
\bibitem{huffman2010fundamentals}
\bibinfo{author}{W~Cary \surnamestart Huffman\surnameend} \&
  \bibinfo{author}{Vera \surnamestart Pless\surnameend} (\bibinfo{year}{2010}):
  \emph{\bibinfo{title}{Fundamentals of error-correcting codes}}.
\newblock \bibinfo{publisher}{Cambridge university press}.

\bibitemdeclare{article}{kalra2019demonstration}
\bibitem{kalra2019demonstration}
\bibinfo{author}{Amolak~Ratan \surnamestart Kalra\surnameend},
  \bibinfo{author}{Navya \surnamestart Gupta\surnameend},
  \bibinfo{author}{Bikash~K \surnamestart Behera\surnameend},
  \bibinfo{author}{Shiroman \surnamestart Prakash\surnameend} \&
  \bibinfo{author}{Prasanta~K \surnamestart Panigrahi\surnameend}
  (\bibinfo{year}{2019}): \emph{\bibinfo{title}{Demonstration of the no-hiding
  theorem on the 5-Qubit IBM quantum computer in a category-theoretic
  framework}}.
\newblock {\sl \bibinfo{journal}{Quantum Information Processing}}
  \bibinfo{volume}{18}(\bibinfo{number}{6}), pp. \bibinfo{pages}{1--13}.

\bibitemdeclare{phdthesis}{karunakaran1975systems}
\bibitem{karunakaran1975systems}
\bibinfo{author}{T~\surnamestart Karunakaran\surnameend}
  (\bibinfo{year}{1975}): \emph{\bibinfo{title}{Systems connections and
  categories}}.
\newblock Ph.D. thesis, \bibinfo{school}{IIT Delhi}.

\bibitemdeclare{article}{kissinger2022phase}
\bibitem{kissinger2022phase}
\bibinfo{author}{Aleks \surnamestart Kissinger\surnameend}
  (\bibinfo{year}{2022}): \emph{\bibinfo{title}{Phase-free ZX diagrams are CSS
  codes (... or how to graphically grok the surface code)}}.
\newblock {\sl \bibinfo{journal}{arXiv preprint arXiv:2204.14038}}.

\bibitemdeclare{article}{PhysRevA.102.022406}
\bibitem{PhysRevA.102.022406}
\bibinfo{author}{Aleks \surnamestart Kissinger\surnameend} \&
  \bibinfo{author}{John \surnamestart van~de Wetering\surnameend}
  (\bibinfo{year}{2020}): \emph{\bibinfo{title}{Reducing the number of
  non-Clifford gates in quantum circuits}}.
\newblock {\sl \bibinfo{journal}{Phys. Rev. A}} \bibinfo{volume}{102}, p.
  \bibinfo{pages}{022406}, \doi{10.1103/PhysRevA.102.022406}.
\newblock \urlprefix\url{https://link.aps.org/doi/10.1103/PhysRevA.102.022406}.

\bibitemdeclare{article}{kissinger2020reducing}
\bibitem{kissinger2020reducing}
\bibinfo{author}{Aleks \surnamestart Kissinger\surnameend} \&
  \bibinfo{author}{John \surnamestart van~de Wetering\surnameend}
  (\bibinfo{year}{2020}): \emph{\bibinfo{title}{Reducing the number of
  non-Clifford gates in quantum circuits}}.
\newblock {\sl \bibinfo{journal}{Physical Review A}}
  \bibinfo{volume}{102}(\bibinfo{number}{2}), p. \bibinfo{pages}{022406}.

\bibitemdeclare{article}{Aleks1}
\bibitem{Aleks1}
\bibinfo{author}{Aleks \surnamestart Kissinger\surnameend} \&
  \bibinfo{author}{John \surnamestart van~de Wetering\surnameend}
  (\bibinfo{year}{2022}): \emph{\bibinfo{title}{Simulating quantum circuits
  with ZX-calculus reduced stabiliser decompositions}}.
\newblock {\sl \bibinfo{journal}{Quantum Science and Technology}}.
\newblock
  \urlprefix\url{http://iopscience.iop.org/article/10.1088/2058-9565/ac5d20}.

\bibitemdeclare{article}{prakash2021normal}
\bibitem{prakash2021normal}
\bibinfo{author}{Shiroman \surnamestart Prakash\surnameend},
  \bibinfo{author}{Amolak~Ratan \surnamestart Kalra\surnameend} \&
  \bibinfo{author}{Akalank \surnamestart Jain\surnameend}
  (\bibinfo{year}{2021}): \emph{\bibinfo{title}{A normal form for single-qudit
  Clifford+ T operators}}.
\newblock {\sl \bibinfo{journal}{Quantum Information Processing}}
  \bibinfo{volume}{20}(\bibinfo{number}{10}), pp. \bibinfo{pages}{1--26}.

\bibitemdeclare{article}{DPS2009}
\bibitem{DPS2009}
\bibinfo{author}{Vishal \surnamestart Sahni\surnameend}, \bibinfo{author}{Dayal
  \surnamestart Srivastava\surnameend} \& \bibinfo{author}{Prem \surnamestart
  Satsangi\surnameend} (\bibinfo{year}{2009}): \emph{\bibinfo{title}{Unified
  modelling theory for qubit representation using quantum field graph models}}.
\newblock {\sl \bibinfo{journal}{Journal of the Indian Institute of Science}}
  \bibinfo{volume}{89}.

\bibitemdeclare{misc}{pawelblog}
\bibitem{pawelblog}
\bibinfo{author}{Pawel \surnamestart Sobocinski\surnameend}:
  \emph{\bibinfo{title}{Graphical Linear Algebra: Orthogonality and
  projections}}.
\newblock
  \bibinfo{howpublished}{\url{https://graphicallinearalgebra.net/2017/08/09/orthogonality-and-projections/}}.

\bibitemdeclare{article}{PhysRevA.75.032110}
\bibitem{PhysRevA.75.032110}
\bibinfo{author}{Robert~W. \surnamestart Spekkens\surnameend}
  (\bibinfo{year}{2007}): \emph{\bibinfo{title}{Evidence for the epistemic view
  of quantum states: A toy theory}}.
\newblock {\sl \bibinfo{journal}{Phys. Rev. A}} \bibinfo{volume}{75}, p.
  \bibinfo{pages}{032110}, \doi{10.1103/PhysRevA.75.032110}.
\newblock \urlprefix\url{https://link.aps.org/doi/10.1103/PhysRevA.75.032110}.

\bibitemdeclare{incollection}{spekkens2016quasi}
\bibitem{spekkens2016quasi}
\bibinfo{author}{Robert~W \surnamestart Spekkens\surnameend}
  (\bibinfo{year}{2016}): \emph{\bibinfo{title}{Quasi-quantization: classical
  statistical theories with an epistemic restriction}}.
\newblock In: {\sl \bibinfo{booktitle}{Quantum Theory: Informational
  Foundations and Foils}}, \bibinfo{publisher}{Springer}, pp.
  \bibinfo{pages}{83--135}.

\bibitemdeclare{article}{DPS2014}
\bibitem{DPS2014}
\bibinfo{author}{D.~\surnamestart Srivastava\surnameend},
  \bibinfo{author}{Vishal \surnamestart Sahni\surnameend} \&
  \bibinfo{author}{P.~\surnamestart Satsangi\surnameend}
  (\bibinfo{year}{2014}): \emph{\bibinfo{title}{Graph-theoretic quantum system
  modelling for neuronal microtubules as hierarchical clustered quantum
  Hopfield networks}}.
\newblock {\sl \bibinfo{journal}{International Journal of General Systems}}
  \bibinfo{volume}{43}, \doi{10.1080/03081079.2014.893298}.

\bibitemdeclare{article}{DPS2011}
\bibitem{DPS2011}
\bibinfo{author}{Dayal \surnamestart Srivastava\surnameend},
  \bibinfo{author}{Vishal \surnamestart Sahni\surnameend} \&
  \bibinfo{author}{Prem \surnamestart Satsangi\surnameend}
  (\bibinfo{year}{2011}): \emph{\bibinfo{title}{Graph-theoretic quantum system
  modelling for information/computation processing circuits}}.
\newblock {\sl \bibinfo{journal}{International Journal of General Systems}}
  \bibinfo{volume}{40}, pp. \bibinfo{pages}{777--804},
  \doi{10.1080/03081079.2011.602016}.

\bibitemdeclare{article}{DPS1}
\bibitem{DPS1}
\bibinfo{author}{Dayal~Pyari \surnamestart Srivastava\surnameend},
  \bibinfo{author}{Vishal \surnamestart Sahni\surnameend} \&
  \bibinfo{author}{Prem~Saran \surnamestart Satsangi\surnameend}
  (\bibinfo{year}{2016}): \emph{\bibinfo{title}{Modelling microtubules in the
  brain as n-qudit quantum Hopfield network and beyond}}.
\newblock {\sl \bibinfo{journal}{International Journal of General Systems}}
  \bibinfo{volume}{45}(\bibinfo{number}{1}), pp. \bibinfo{pages}{41--54},
  \doi{10.1080/03081079.2015.1076405}.
\newblock \urlprefix\url{https://doi.org/10.1080/03081079.2015.1076405}.

\bibitemdeclare{article}{DPS2}
\bibitem{DPS2}
\bibinfo{author}{Dayal~Pyari \surnamestart Srivastava\surnameend},
  \bibinfo{author}{Vishal \surnamestart Sahni\surnameend} \&
  \bibinfo{author}{Prem~Saran \surnamestart Satsangi\surnameend}
  (\bibinfo{year}{2017}): \emph{\bibinfo{title}{From n-qubit multi-particle
  quantum teleportation modelling to n-qudit contextuality based quantum
  teleportation and beyond}}.
\newblock {\sl \bibinfo{journal}{International Journal of General Systems}}
  \bibinfo{volume}{46}(\bibinfo{number}{4}), pp. \bibinfo{pages}{414--435},
  \doi{10.1080/03081079.2017.1308361}.
\newblock \urlprefix\url{https://doi.org/10.1080/03081079.2017.1308361}.

\bibitemdeclare{article}{weinstein1987symplectic}
\bibitem{weinstein1987symplectic}
\bibinfo{author}{Alan \surnamestart Weinstein\surnameend}
  (\bibinfo{year}{1987}): \emph{\bibinfo{title}{Symplectic groupoids and
  Poisson manifolds}}.
\newblock {\sl \bibinfo{journal}{Bulletin (New Series) of the American
  Mathematical Society}} \bibinfo{volume}{16}(\bibinfo{number}{1}), pp.
  \bibinfo{pages}{101--104}.

\bibitemdeclare{article}{Wootters1987}
\bibitem{Wootters1987}
\bibinfo{author}{W.~K. \surnamestart {Wootters}\surnameend}
  (\bibinfo{year}{1987}): \emph{\bibinfo{title}{{A Wigner-function formulation
  of finite-state quantum mechanics}}}.
\newblock {\sl \bibinfo{journal}{Annals of Physics}} \bibinfo{volume}{176}, pp.
  \bibinfo{pages}{1--21}, \doi{10.1016/0003-4916(87)90176-X}.

\end{thebibliography}

\end{document}